\documentclass[11pt]{article}
\usepackage{fullpage}

\usepackage{cite}
\usepackage{fancyvrb}
\DefineVerbatimEnvironment{code}{Verbatim}{fontsize=\small}
\DefineVerbatimEnvironment{example}{Verbatim}{fontsize=\small}
\newcommand{\ignore}[1]{}
\usepackage{algorithmic}          
\usepackage{algorithm}
\usepackage{amsthm}
\usepackage{amsmath}
\usepackage{amssymb}
\usepackage{color}
\usepackage{graphicx}
\usepackage{fancyhdr}
\usepackage{subfig}
\usepackage{cite}
\usepackage{epstopdf}
\usepackage{verbatim}
\usepackage[iso-8859-7]{inputenc}
\usepackage[pdftex]{hyperref}
\hypersetup{
    colorlinks,%
    citecolor=black,%
    filecolor=black,%
    linkcolor=black,%
    urlcolor=black
}

\newtheorem{theorem}{Theorem}
\newtheorem{corollary}[theorem]{Corollary}
\newtheorem{lemma}[theorem]{Lemma}
\newtheorem{property}[theorem]{Property}
\newtheorem*{remark2}{Estimating $x_{1,1}$ and $x_{2,1}$}

\newtheorem{definition}[theorem]{Definition}

\newtheorem{fact}[theorem]{Fact}

\newtheorem{mygame}{Game}

\def\calG{\mathcal{G}} \def\xx{\mathbf{x}} \def\kk{\mathbf{k}} \def\ss{\mathbf{s}} \def\k{{K}} \def\pay{\textsf{payoff}}  \def\bxx{{\mathcal{X}}} \def\yy{\mathbf{y}}
\def\Pr{\text{Pr}} \def\tt{\mathbf{t}} \def\calX{\mathcal{X}} \def\AA{\mathbf{A}}
\def\yy{\mathbf{y}} \def\calL{\mathcal{L}} \def\calS{\mathcal{S}} \def\calT{\mathcal{T}}
\def\colorful{1} \def\byy{\mathcal{Y}} \def\BB{\mathbf{B}} \def\CC{\mathbf{C}}

\ifnum\colorful=1

\fi
\ifnum\colorful=0

\fi

\begin{document}

\title{On the Complexity of Nash Equilibria in Anonymous Games}
\author{Xi Chen \\ Columbia University 
\and David Durfee \\ Georgia Institute of Technology  
\and Anthi Orfanou \\ Columbia University} 
\date{}
\maketitle 

\begin{abstract}
We show that the problem of finding an $\epsilon$-approximate Nash equilibrium
  in an {anonymous} game with seven pure strategies is complete in PPAD, when  
  the approximation parameter $\epsilon$ is exponentially small in the number of players.
\end{abstract}

\section{Introduction}

The celebrated theorem of Nash \cite{NASH50, NASH51} states that every game has an equilibrium point. 
The concept of Nash equilibrium has been tremendously influential 
  in economics and social sciences ever since (e.g., see \cite{NashSocial}),
  and its computation has been one the 
  most well-studied problems in the area of Algorithmic Game Theory. 
For {normal form games} with a bounded number of players,
  much progress has~been made during the past decade in 
  understanding both the complexity of Nash equilibrium 
  \cite{AbbottKaneValiant,ChenDengTengSparse, ChenTengValiant,DGPJournal, 2Nash,  Mehta14}
as well as its efficient approximation \cite{Lipton,BaranyVempalaVetta,KontogiannisPana,DaskalakisMehtaPapa,  
   DaskalakisMehtaPapa2, BosseMarkakis, KannanTheobald,TSepsilon, DederNazerzadehSaberi,
  KontogiannisSpirakis,PlanarP,TsaknakisS}.
  
In this paper we study a large and important 
  class of \emph{succinct multiplayer} games  called \emph{anonymous games} (see   
  \cite{sch73, Mil96, Blo99, Blo05, Kal05} for studies of such games in the economics~literature). 
These are special multiplayer games in that the payoff of each player depends only on 
  (1)~the pure strategy of the player herself, and 
  (2) the \emph{number} of other players playing each pure strategy,
  instead of the full pure strategy profile.
In such a game, the (expected) payoff of a player is \emph{highly symmetric}
  over (pure or mixed) strategies of other players.
For instance, two players switching~their strategies would not affect the payoff of any other player.
A consequence of this very special payoff structure is that 
  $O(\alpha n^{\alpha-1})$ numbers suffice to completely
  describe the payoff function of a player, when 
  there are $\alpha$ pure strategies shared by $n$ players. 
Notably this is polynomial in the number of players
  when $\alpha$ is bounded, and hence the game is {succinctly representable}.
\emph{Throughout the paper, we focus on succinct anonymous games with a bounded number 
  of pure strategies.}

Other well-studied multiplayer games with a succinct representation include 
  graphical, symme\-tric, and congestion games (for more details see \cite{PapRough08}).   
While graphical and congestion games are both known to be hard to solve 
  \cite{fabrikant04,AckermannPLS,Skopalik},
  there is indeed a polynomial-time algorithm for computing an exact Nash equilibrium
  in a symmetric game \cite{PapRough08}.  
Because anonymous games generalize symmetric games 
  by allowing player-dependent payoff functions,
it is a natural question to ask whether there is an efficient algorithm 
  for finding an (exact or approximate) Nash equilibrium in an anonymous game as well.

Culminating in a sequence of beautiful papers \cite{PTAS07, 2PTAS08,cPTAS08,obliv, dask14}
  Daskalakis and Papadimitriou obtained a polynomial-time approximation scheme
  (PTAS) for $\epsilon$-approximate Nash equilibria in anonymous games with a bounded number of strategies
  (see more discussion on related work in Section \ref{sec:work}).
However, the complexity of finding an exact Nash equilibrium in such games
  remains open, and was conjectured to be hard for PPAD 
in \cite{obliv,dask14}.\hspace{0.06cm}\footnote{When the 
number of pure strategies is a sufficiently large constant, an anonymous 
  game with rational payoffs may not have any rational equilibrium (e.g., 
    by embedding in it a rational three-player game with no rational equilibrium).
But for the case of two strategies, it remains unclear as whether every 
  rational anonymous game has a rational Nash equilibrium, which was posed as
  an open problem in \cite{dask14}.}

\def\anonymous{\textsc{Anonymous}}

In this paper we give an affirmative answer to the conjecture of Daskalakis and Papadimitriou,
  by showing that it is PPAD-complete to find an $\epsilon$-approximate Nash equilibrium
  in an anonymous game, when the approximation parameter $\epsilon$ is exponentially small in $n$.
To formally state our main result, let $(\alpha,c)$-\anonymous\ denote the problem of finding
  a $(2^{-n^c})$-approximate Nash equilibrium in an anonymous game
  with $\alpha$ pure strategies and payoffs from $[0,1]$.\hspace{0.06cm}\footnote{Since 
  we are interested in the additive approximation,
  all payoffs are normalized to take values in $[0,1]$.}

Here is our main theorem: 

\begin{theorem}\label{main-theorem}
For any $\alpha\ge 7$ and $c>0$, the problem
  $(\alpha,c)$-\anonymous\ is PPAD-complete.
\end{theorem}

\emph{The greatest challenge to establishing the PPAD-completeness result
  stated above
  is 
  posed by the rather complex but also highly symmetric 
  payoff structure of anonymous games.}
Before discussing our approach and techniques in Section \ref{sec:ours},
  we first review related work in Section \ref{sec:work}, then
  define anonymous games formally and introduce some useful notation in Section \ref{sec:pre}.

\subsection{Related Work}\label{sec:work}

Anonymous games have been studied extensively in the economics literature 
  \cite{sch73, Rashid83, Mil96, Blo99, Blo05, Kal05, Sandholm05}, 
where the game being considered is usually nonatomic and consists~of~a 
  continuum of players but a finite number of strategies.
For the discrete setting, two~special families of 
  anonymous games are symmetric games \cite{PapRough08,Bra09}
  and congestion games \cite{ros}.
\cite{PapRough08} gave a polynomial-time for finding an exact Nash equilibrium in a symmetric
  game.
For congestion games, PLS-completeness of 
  pure equilibria was established in 
  \cite{fabrikant04,AckermannPLS,Skopalik}\hspace{0.03cm}\footnote{These
  PLS-hardness results have no implication to the setup of this paper 
  since the number of pure strategies in the congestion
  games considered there are unbounded.},
and efficient approximation algorithms for various 
  latency functions were obtained in \cite{CFGS1,CFGS2,Chien07}.

While an anonymous game does not possess a pure Nash equilibrium in general,
  it was shown in \cite{PTAS07, Azrieli, dask14} that
  when the payoff functions are $\lambda$-Lipschitz,
there exists an $\epsilon$-approximate pure Nash equilibrium
  and it can be found in polynomial time, where
  $\epsilon$ has a linear dependency on $\lambda$.
Furthermore, in \cite{BRE12} Babichenko presented a best-reply dynamic 
  for $\lambda$-Lipschitz anonymous games with two strategies 
  which reaches an approximate pure equilibrium in $O(n\log n)$ steps.

Regarding our specific point of interest, i.e., (mixed) Nash equilibria
  in anonymous games with 
  a scaling number of players but a non-scaling number of strategies,  
  there have been a sequence of positive and negative results obtained by
  Daskalakis and Papadimitriou \cite{PTAS07, cPTAS08,2PTAS08, obliv} (summarized
  in the journal version \cite{dask14}). 
We briefly review these results below.

In \cite{PTAS07}, Daskalakis and Papadimitriou 
  presented a PTAS for finding an $\epsilon$-approximate Nash equilibrium
  in an anonymous game with two pure strategies,
  with running time $n^{O(1/\epsilon^2)}\cdot U$, where $U$ denotes
  the number of bits required to describe the payoffs. 
The running time was subsequently improved in \cite{2PTAS08} to 
  $\text{poly}(n)\cdot (1/\epsilon)^{O(1/\epsilon^2)}\cdot U$. 
The first PTAS in \cite{PTAS07} is based on the existence of an $\epsilon$-approximate
  Nash equilibrium consisting of integer multiples of $\epsilon^2$,
  while the second PTAS in \cite{2PTAS08} is
  based on the existence of an $\epsilon$-approximate Nash equilibrium
  satisfying the following special property: either at most $O(1/\epsilon^3)$ players 
  play mixed strategies, or all players who mix~play the same mixed strategy.
Later \cite{cPTAS08} extended the result 
  of \cite{PTAS07}, giving the only known PTAS for anonymous games with 
  any bounded number of pure strategies with time 
  $n^{g(\alpha,1/\epsilon)}\cdot U$ for some function $g$ of $\alpha$,
  number of pure strategies, and $1/\epsilon$.

All three PTAS obtained in \cite{PTAS07,2PTAS08,cPTAS08} are
  so-called \emph{oblivious} algorithms \cite{obliv},~i.e.,
  algorithms that enumerate a set of mixed strategy profiles that is independent of the 
  input game as candidates for approximate Nash equilibria
  (hence, the game is used only to verify if a given mixed strategy profile
  is an $\epsilon$-approximate Nash equilibrium).
In \cite{obliv}, Daskalakis and Papadimitriou showed that 
  any oblivious algorithm for anonymous games must have running
  time exponential in $1/\epsilon$.
In contrast \vspace{-0.036cm}to this negative result, they 
  also presented a {\it non-oblivious} PTAS for two-strategy anonymous games  
  with running time $\text{poly}(n)\cdot (1/\epsilon)^{O(\log^2 (1/\epsilon))} \cdot U$. \vspace{-0.04cm}

\subsection{Anonymous Games and Polymatrix Games\vspace{-0.06cm}}\label{sec:pre}

Before giving a high-level description of our approach
  and techniques in Section \ref{sec:ours}, we first give a formal definition of
  anonymous games and introduce some useful notation. 
Consider a multiplayer game with $n$ players  $[n]=\{1,\ldots,n\}$ 
  and $\alpha$ pure strategies $[\alpha]=\{1,\ldots,\alpha\}$ with $\alpha$ being a constant.
For each pure strategy $b\in [\alpha]$, let
  $\psi_b(\tt)$ denote the number of $b$'s in a tuple $\tt\in [\alpha]^{n-1}$,
and define $\Psi(\tt) = (\psi_{1}(\tt),\ldots, \psi_{\alpha}(\tt)),$
  which we will refer to as the \emph{histogram} of pure strategies in $\tt$.

In an anonymous game, the payoff of each player $p\in [n]$~depends only on 
  $\Psi(\ss_{-p})$ and her own strategy $s_p$, given a pure strategy profile $\ss\in [\alpha]^n$. 
(We follow the convention
  and use $\ss_{-p}\in [\alpha]^{n-1}$ to denote the 
  pure strategy profile of the $n-1$ players other than player $p$ in $\ss$.)
Informally, $\Psi(\ss_{-p})$ can be described as 
  what player $p$ ``sees''   in the game when $\ss$ is played. 

We now formally define {anonymous games}.

\begin{definition}
An {anonymous game} $\calG = (n,\alpha,\{\textsf{\emph{payoff}}_p\})$ consists of a set $[n]$ of $n$ players, a set~$[\alpha]$ of $\alpha$ pure strategies, and a payoff function $\emph{\pay}_p : [\alpha]\times
  \k \rightarrow \mathbb{R}$ for each player $p\in [n]$, where 
\begin{eqnarray*}
&\k =  \big\{(k_1,\ldots,k_\alpha): k_j \in \mathbb{Z}_{\ge 0}\text{ for all $j$ 
  and }\sum_{j=1}^\alpha k_j = n-1 \big\}&
\end{eqnarray*} 
is the set of all histograms of pure strategies played by $n-1$ players.
Specifically, when $\ss\in [\alpha]^{n}$ is played,
  the payoff of player $p$ is given by $\textsf{\emph{payoff}}_p(s_p,\Psi(\ss_{-p}))$.
\end{definition}

As usual, a {mixed strategy} is a probability distribution 
  $\xx=(x_{1},\ldots,x_{\alpha})$, and
  a {mixed strategy profile} $\bxx$ is an ordered tuple of 
  $n$ mixed strategies $(\xx_p:p\in [n])$, one for each player $p$.
Given $\bxx$, 
  let $u_p(b,\bxx)$ denote the \emph{expected payoff} of 
  $p$ playing $b\in [\alpha]$, which has the 
  following explicit expression:
$$
u_p(b,\bxx) = \sum_{\kk \in \k}{\pay}_p(b,\kk) \cdot \Pr_{\bxx}[p,\kk],
$$
where $\Pr_{\bxx}[p,\kk]$ denotes the probability of player $p$ 
  seeing histogram $\kk$ under $\bxx$:
$$
\Pr_{\bxx}[p,\kk] = \sum_{\ss_{-p} \in \Psi^{-1}(\kk)} \left(
\hspace{0.05cm}\prod_{q\ne p} x_{q,s_q}\right). 
$$
Note that $s_q$ denotes the pure strategy of player $q$
  from a profile $\ss_{-p}\in \Psi^{-1}(\kk)$.
We also use $u_p(\bxx)$ to
  denote the expected payoff of player $p$ from playing $\xx_p$:
$$u_p(\bxx)=\sum_{b\in [\alpha]} x_{p,b}\cdot u_p(b,\calX).$$
It is worth pointing out that, while $u_p(b,\bxx)$ contains exponentially many terms, it can be 
  computed in polynomial time using dynamic programming \cite{PTAS07,dask14} when 
  $\alpha$ is a constant.
For a detailed presentation of the algorithm
  for $2$-strategy anonymous games, see \cite{dask14}.
This then implies that checking whether a given profile $\calX$
  is a (approximate) Nash equilibrium is in polynomial time.

Next we define (approximate) Nash equilibria of an anonymous game.

\begin{definition}
Given an anonymous game $\calG=(n,\alpha,\{\textsf{\emph{payoff}}_p\})$,
  we say a mixed strategy profile $\calX$ is a 
  \emph{Nash equilibrium} of $\calG$ if $u_p(\calX)\ge u_p(b,\calX)$
  for all players $p\in [n]$ and strategies $b\in [\alpha]$.

For $\epsilon\ge 0$, we say $\calX$ is an \emph{$\epsilon$-approximate Nash equilibrium} if
  $u_p(\calX)+\epsilon\ge u_p(b,\calX)$ for all $p\in [n]$ and $b\in [\alpha]$.
For $\epsilon\ge 0$, we say $\calX$ is an \emph{$\epsilon$-well-supported Nash equilibrium} if
  $u_p(a,\calX)+\epsilon<u_p(b,\calX)$ implies that 
  $x_{p,a}=0$, for all $p\in [n]$ and $a,b\in [\alpha]$.
\end{definition}

As discussed in Section \ref{sec:ours}, the hardness part of Theorem \ref{main-theorem} 
  is proved using a polynomial-time reduction from the problem of finding
  a well-supported Nash equilibrium in a \emph{polymatrix game} (e.g. see \cite{Cai11}).
For our purposes, such a game (with $n$ players and two strategies each player)
  can be described by a payoff matrix $\AA \in [0,1]^{2n\times 2n}$ 
  with $A_{k,\ell}=0$ for all $k ,\ell \in \{2i-1,2i\}$ and $i \in [n]$. 

Each player $i\in [n]$ has two pure strategies that correspond to rows $2i-1$ and $2i$ of $\AA$. 
Let~$\AA_{j}$ denote the $j$th row of $\AA$.
Given a vector $\yy\in \mathbb{R}_{\ge 0}^{2n}$,
  where $(y_{2i-1},y_{2i})$ is the mixed strategy of player $i$, expected payoffs
  of player $i$ for playing rows $2i-1$ and $2i$ are $\AA_{2i-1}\cdot \yy$
  and $\AA_{2i}\cdot \yy$ respectively.

An \emph{$\epsilon$-well-supported Nash equilibrium} of $\AA$ is a vector
  $\yy\in \mathbb{R}_{\ge 0}^{2n}$ such that $y_{2i-1}+y_{2i}=1$ and  
\begin{align*}
\AA_{2i-1} \cdot \yy  > \AA_{2i} \cdot {\yy} +\epsilon
\hspace{0.06cm} \Rightarrow\hspace{0.06cm}  y_{2i}=0
\ \ \  \text{and}\ \ \ 
\AA_{2i}\cdot  \yy  > \AA_{2i-1}\cdot \yy +\epsilon
\hspace{0.06cm} \Rightarrow\hspace{0.06cm}  y_{2i-1}=0,
\end{align*}
for all players $i\in [n]$.
We need the following result on such games:

\begin{theorem}[\hspace{0.0001cm}\cite{market}]
The problem of computing a $(1/n)$-well-supported Nash equilibrium
  in a\\ polymatrix game is PPAD-complete. \vspace{-0.05cm}
\end{theorem}

\subsection{Our Approach and Techniques\vspace{-0.06cm}}\label{sec:ours}

A commonly used approach to establishing the PPAD-hardness of approximate 
  equilibria is to~design gadget games that can perform 
  certain arithmetic operations on entries of mixed strategies of players 
  (e.g. see \cite{DGPJournal,2Nash}). 
Such gadgets would then yield a reduction from 
  the problem of solving a generalized circuit \cite{DGPJournal,2Nash}, a
  problem  complete in PPAD. 
However, we realized that this approach may not work well with anonymous games; 
  we found that it was impossible to design an anonymous game $G_{=}$ that enforces 
  equality constraints.\hspace{0.03cm}\footnote{For example, 
  we can rule out the existence of an anonymous game $G_{=}$
  with 4 players and 2 pure strategies such that $\xx$ is a Nash equilibrium of $G_{=}$ if and only if $x_{1}=x_{2} \in [\mu,\nu] \subseteq [0,1]$ and $x_{3}=x_{4} \in [\mu',\nu'] \subseteq [0,1]$, where 
  we use $x_i$ to denote the probability that player $i$ plays the first pure strategy.}

Instead we show the PPAD-hardness of anonymous games via 
  a reduction from the problem~of finding a $(1/n)$-well-supported
  equilibrium in a two-strategy polymatrix game (see
  Section \ref{sec:pre}).
Given a $2n\times 2n$ polymatrix game $\AA$,
  our reduction constructs an anonymous game $\calG_\AA$ with $n$ ``main'' players $\{P_1,\ldots,P_n\}$
  (and two auxiliary players).
We have each main player $P_i$ simulate in a way a player $i$ in the polymatrix game,
  as discussed below, such that   
  any $\epsilon$-well-supported Nash equilibrium of $\calG_\AA$ with an exponentially 
  small $\epsilon$ can be used to recover
  a $(1/n)$-well-supported Nash equilibrium of the polymatrix game $\AA$ efficiently. We then prove a connection between approximate Nash equilibria and well-supported Nash
  equilibria of anonymous games to finish the proof of Theorem \ref{main-theorem}.

\emph{The greatest challenge to establishing such a reduction
  is posed by the complex but highly structured, symmetric expression of
  expected payoffs in an anonymous game.}
As discussed previously~in Section \ref{sec:pre}, the expected payoff $u_p(b,\calX)$ 
  of player~$p$ is a linear form of probabilities $\Pr_{\cal{X}}[p,\kk]$,
  each of which 
  is function over mixed strategies of all players other than $p$.
This rather complex function makes it difficult to reason about  
  the set of well-supported Nash equilibria of an anonymous game,
  not to mention our goal is to embed a polymatrix game in it.
To overcome this obstacle, we need to find a special (but hard enough) family of 
  anonymous games with certain payoff structures which allow us to 
  perform a careful analysis and understand their well-supported equilibria.
The bigger~obstacle for our reduction, however,
  is to in some sense \emph{remove the anonymity of the players and break
  the inherent symmetry underlying an anonymous game}.

To see this, a natural approach to obtain a reduction from polymatrix games
  is to directly encode the $2n$ variables of $\yy$ in mixed strategies of the $n$ ``main'' players $\{P_1,\ldots,P_n\}$. More~specifically, let $\{s_1, s_2\}$ denote two special pure strategies of $\calG_\AA$, and we attempt to encode $(y_{2i-1}, y_{2i})$ in $(x_{i,s_1} , x_{i,s_2} )$, probabilities of $P_i$ playing $s_1,s_2$, respectively. The reduction would work if expected payoffs of $P_i$ from $s_1$ and $s_2$ in $\calG_\AA$ can always match closely expected payoffs of player $i$ from~rows $2i - 1$ and $2i$ in $\AA$, given by two linear forms $\AA_{2i-1}\cdot \yy$ and 
   $\AA_{2i}\cdot \yy$ of $\yy$. However, it seems difficult, if not impossible, to construct $\calG_\AA$ with this property, since anonymous games are highly symmetric: the expected payoff of $P_i$ is a symmetric function over mixed strategies of all other players. This~is not the case for polymatrix games: a linear form such as $\AA_{2i}\cdot \yy$ in general has different coefficients for different variables, so different players contribute with different weights to the expected payoff of a player (and the problem of finding a well-supported equilibrium
   in $\AA$ clearly becomes trivial if we require that every row of 
   $\AA$ has the same entry).
 
An alternative approach is to encode the $2n$ variables of $\yy$
   in probabilities $\Pr_{\cal{X}}[p,\kk]$. This may look appealing because expected payoffs
   $u_p(b,\calX)$ are linear forms of these probabilities
so one can set the coefficients $\pay_p(b,\kk)$ to match 
  them easily with those linear forms
  $\AA_{j}\cdot \yy$ that appear in the polymatrix game $\AA$. However,  
  the histogram $\kk$ seen by a player $p$ (as a vector-valued random variable)
  is~the sum of $n-1$ vector-valued random variables,
  each distributed according to the~mixed strategy of a player other than $p$.
The way these probabilities $\Pr_{\cal{X}}[p,\kk]$ are derived in turn imposes 
  strong restrictions on them,\hspace{0.02cm}\footnote{For example, as it is pointed out in \cite{PTAS07,cPTAS08} for anonymous games with two strategies, players can always be partitioned into a few sets such that the probabilities $\Pr_{\cal{X}}[p,\kk]$ over $\kk$ must follow approximately a Poisson or a discretized Normal distribution on  each set respectively.}
  which makes it a difficult task to obtain a correspondence between the $2n$ free variables in $\yy$ and the probabilities $\Pr_{\cal{X}}[p,\kk]$.

Our reduction indeed follows the first approach of encoding $(y_{2i-1},y_{2i})$
  in $(x_{i,s_1},x_{i,s_2})$ of player $P_i$.
More exactly, the former is the normalization of the latter
  into a probability distribution.
Now to overcome the difficulty posed by symmetry, we \emph{enforce} the following
  \emph{``scaling'' property} in every well-supported Nash equilibrium $\calX$
  of $\calG_\AA$: probabilities of $P_i$ playing $\{s_1,s_2\}$ satisfy 
\begin{equation}\label{scales}
x_{i,s_1}+x_{i,s_2}\approx 1/N^i,
\end{equation} where $N$ is exponentially large in $n$.
This property is established
  by designing an anonymous game called \emph{generalized radix game} $\calG_{n,N}^*$,
  and then using it as the base game in the construction of $\calG_\AA$.
We show that (\ref{scales}) 
  holds approximately for every anonymous game that is payoff-wise
  \emph{close} to $\calG_{n,N}^*$.
In particular, (\ref{scales}) holds for any well-supported
  equilibrium of $\calG_\AA$, as long as we make sure $\calG_\AA$ is close to $\calG_{n,N}^*$.
The ``scaling'' property plays a crucial role in our 
  reduction because, as the base game for $\calG_\AA$,
  it helps us reason about well-supported Nash equilibria of $\calG_\AA$;
  it also removes anonymity of the $n$ ``main'' players $P_i$
  (since they must play the two special pure strategies $\{s_1,s_2\}$ 
  with probabilities of different scales) and
  overcome the symmetry barrier.

Equipped with the ``scaling'' property (\ref{scales}), we prove 
  a key technical lemma called the \emph{estimation lemma}.
It shows that one can compute efficiently coefficients of a linear form over probabilities~of 
  histograms $\Pr_\calX[P_i,\kk]$ seen by player $P_i$, which guarantees to approximate additively $x_{j,s_1}$ (or $x_{j,s_2}$) i.e. probability of another player $P_j$ plays $s_1$ (or $s_2$), whenever the profile $\calX$ satisfies the ``scaling'' property
  (this holds when $\calG_\AA$ is close to $\calG_{n,N}^*$
  and $\calX$ is a well-supported equilibrium of $\calG_\AA$).
As $$(y_{2j-1},y_{2j})\approx N^j(x_{j,s_1},x_{j,s_2})$$ given (\ref{scales}),
  these linear forms for $x_{j,s_1},x_{j,s_2}$ 
  can be combined to derive a linear form of
  $\Pr_\calX[P_i,\kk]$~to approximate additively any 
  linear form of $\yy$, particularly $\AA_{2i-1}\cdot \yy$ or 
  $\AA_{2i}\cdot \yy$ that appear~as expected payoffs of player $i$ in the polymatrix game $\AA$.
The proof of the estimation lemma is the technically most 
  involved part of the paper.
We indeed derive explicit expressions for coefficients
  of the desired linear form where substantial cancellations
  yield an additive approximation of $x_{j,s_1}$ or $x_{j,s_2}$.  

Finally we combine all ingredients highlighted above to construct an anonymous
  game $\calG_\AA$ from polymatrix game $\AA$.
This is done by first using the estimation lemma to compute,
  for each~main~$P_i$
  coefficients of linear forms of probabilities 
  $\Pr_\calX[P_i,\kk]$ seen by $P_i$ that yield
  additive approximations of $x_{j,s_1}$ and $x_{j,s_2}$.
\hspace{-0.05cm}We then perturb payoff functions of players $P_i$ 
  in the generalized radix~game $\calG_{n,N}^*$
  using these coefficients  
  so that 1) the resulting game $\calG_\AA$ is close to $\calG_{n,N}^*$
  and thus, any well-supported equilibrium $\calX$ of $\calG_\AA$ 
  automatically satisfies the ``scaling'' property;
  2) expected payoffs of $P_i$ playing $s_1,s_2$ in 
  a well-supported equilibrium $\calX$ of $\calG_\AA$
  match additively expected payoffs of player $i$ playing
  rows $2i-1,2i$ in $\AA$, given $\yy$ derived from $\calX$ by normalizing
  $(x_{j,s_1},x_{j,s_2})$ for each $j$.
The correctness of the reduction, i.e., $\yy$ is a 
  $(1/n)$-well-supported equilibrium of $\AA$ whenever
  $\calX$ is an $\epsilon$-well-supported equilibrium of $\calG_\AA$
  with an exponentially small $\epsilon$, follows from these properties of $\calG_\AA$.

\subsection{Organization\vspace{-0.06cm}}

In Section 2, we define the radix game, and show that it has a unique Nash equilibrium
  as a warm-up.
We also use it to
  define the generalized radix game which serves as the base of our reduction. 
In section 3, we characterize well-supported Nash equilibria of 
   anonymous games that are close to the generalized radix game (i.e., those
   that can be obtained by adding small perturbations to~payoffs of 
  the generalized radix game).
In section 4, we prove the PPAD-hardness part of the main theorem.
Our reduction relies on a crucial technical lemma, called the estimation lemma, which we prove in Section 5.
We prove the membership in Section 6, and conclude with open problems in Section 7.

\section{Warm-up: Radix Game}

In this section, we first define a 
  $(n+2)$-player anonymous game $\calG_{n,N}$, called the \emph{radix game}.
As~a warmup for the next section, we show that it has a unique Nash equilibrium.
We then use the radix game to define the \emph{generalized radix game} $\calG_{n,N}^*$,
  by making a duplicate of a pure strategy in $\calG_{n,N}$.
The latter will serve as the base game for 
  our polynomial-time reduction from polymatrix games.
  
\subsection{Radix Game}
  
The radix game $\calG_{n,N}$ to be defined has a unique Nash equilibrium of a specific form:
given $N\ge 2$ as an integer parameter of the game, each of the $n$ ``main'' 
  players mixes over the first two strategies with probabilities $1/N^i$ 
  and $1-1/N^i$, respectively, for each $i\in [n]$, in the unique
  Nash equilibrium. 
The remaining two ``special'' players are created to achieve the 
  aforementioned property. 

\begin{mygame}[Radix Game $\calG_{n,N}$]\label{gadget1}
Let $n\ge 1$ and $N\ge2$ denote two integer parameters. Let $\delta=1/N$.

Let $\calG_{n,N}$ denote the following anonymous game 
  with $n+2$ players $\{\hspace{-0.03cm}P_1,\ldots,P_n,Q,R\}$ and $6$
  pure strategies $\{s,t,q_1,q_2,r_1,r_2\}$.
We refer to $\{P_1,\ldots,P_n\}$ as the \emph{main} players.
Each main player $P_i$ is only interested in strategies $s$ and $t$ \emph{(}e.g., by setting
  her payoff of playing any other four actions to be $-1$ no matter what other players play\emph{)}.  
Player $Q$ is only interested in strategies $\{q_1,q_2\}$, and
player $R$ is only interested in strategies $\{r_1,r_2\}$.

Next we define the payoff function of each player.
When describing the payoff of a player below we 
  always use $\kk=(k_s,k_t,k_{q_1},k_{q_2},k_{r_1},k_{r_2})$ to denote the histogram
  of strategies this player sees.\vspace{0.03cm}
\begin{flushleft}\begin{enumerate}
\item For each $i\in [n]$, 
the payoff of player $P_i$ when she plays $s$ only depends on $k_s$: \vspace{0.04cm}
$$
\text{\emph{\textsf{payoff}}}_{P_i}(s,\kk)= 
\begin{cases}
    \hspace{0.06cm} \delta^i+ \prod_{j\in [n]} \delta^j&  
    \text{if $k_s=n-1$}\\[0.9ex]
    \hspace{0.06cm} \prod_{j\in [n]} \delta^j & \text{otherwise.}
\end{cases}\vspace{0.04cm}
$$
The payoff of player $P_i$ when she plays $t$ only depends on $k_{r_1}$:
$$
\text{\emph{\textsf{payoff}}}_{P_i}(t,\kk)=\begin{cases}
    \hspace{0.06cm}2& \text{if $k_{r_1}= 1$}\\[0.4ex]
    \hspace{0.06cm}0 & \text{otherwise.}
\end{cases}
$$
\item The payoff of player $Q$ when she plays $q_1$ or $q_2$ is given by
$$
\text{\emph{\textsf{payoff}}}_Q(q_1,\kk)=\begin{cases}
    \hspace{0.06cm}1& \text{if $k_{s}= n$}\\[0.4ex]
    \hspace{0.06cm}0 & \text{otherwise}
\end{cases}
\ \ \ \text{and}\ \ \ \
\text{\emph{\textsf{payoff}}}_Q(q_2,\kk)=\begin{cases}
    \hspace{0.06cm}1& \text{if $k_{r_1}= 1$}\\[0.4ex]
    \hspace{0.06cm}0 & \text{otherwise.}
\end{cases}
$$

\item The payoff of player $R$ when she plays $r_1$ or $r_2$ is given by  
$$
\text{\emph{\textsf{payoff}}}_R(r_1,\kk)=\begin{cases}
    \hspace{0.06cm}1& \text{if $k_{q_1}=1$}\\[0.4ex]
    \hspace{0.06cm}0 & \text{otherwise}
\end{cases}
\ \ \ \text{and}\ \ \ \
\text{\emph{\textsf{payoff}}}_R(r_2,\kk)=\begin{cases}
    \hspace{0.06cm}1& \text{if $k_{q_2}=1$}\\[0.4ex]
    \hspace{0.06cm}0 & \text{otherwise.}
\end{cases}
$$
\end{enumerate}
This finishes the definition of the radix game $\calG_{n,N}$.
\end{flushleft}
\end{mygame}

\begin{fact} 
$\calG_{n,N}$ is an anonymous game
  with payoff functions taking values from $[-1,2]$.
\end{fact}

Since the main players $P_i$ are only interested in $\{s,t\}$,
  $Q$ is only interested in $\{q_1,q_2\}$, and $R$~is only interested in $\{r_1,r_2\}$,
  each Nash equilibrium $\calX$ of   
  $\calG_{n,N}$ can be fully specified by a $(n+2)$-tuple 
  $\calX=(x_1,\ldots,x_n,y,z)\in [0,1]^{n+2}$, where
  $x_i$ denotes the probability of $P_i$ playing strategy $s$ for
  each $i\in [n]$, $y$ denotes the probability of $Q$ playing $q_1$,
  and $z$ denotes the probability of $R$ playing $r_1$.

Given $\calX=(x_1,\ldots,x_n,y,z)$
  we calculate the expected payoff of each player as follows
  (we skip $\calX$ in the expected payoffs $u_p(b,\calX)$, when 
  $\calX$ is clear from the context, and we use $u_i$ to denote 
  the expected payoff of $P_i$ instead of $u_{P_i}$ for convenience):

\begin{fact}\label{fact1}
Given $\calX=(x_1,\ldots,x_n,y,z)$, the expected payoff of player 
  $P_i$ for playing $s$ is \vspace{0.03cm}
$$u_{i}(s)=\delta^i\cdot \emph{\Pr}\big[k_s=n-1\big] + \prod_{j \in [n]} \delta^j = \delta^i \prod_{j\ne  i \in [n]} x_j + \prod_{j \in [n]} \delta^j.$$ 
The expected payoff of $P_i$ for playing $t$ is $u_{i}(t)=2z$. 

The expected payoff of player $Q$ for playing $q_1$ is\vspace{0.03cm}
$$u_{Q}(q_1)= \emph{\Pr}\big[k_s=n\big]=\prod_{i\in [n]} x_i.\vspace{-0.03cm}$$
The expected payoff of $Q$ for playing $q_2$ is 
$u_Q(q_2)=z$. 

The expected payoff of $R$ for playing $r_1$ is $u_R(r_1)=y$
  and that for $r_2$ is $u_R(r_2)=1-y$.  
\end{fact}

We show that $x_i=\delta^i$ in a Nash equilibrium $\calX$ of $\calG_{n,N}$.
We start with the following lemma.

\begin{lemma} \label{notuseful1}
In a Nash equilibrium $\calX=(x_1,\ldots,x_n,y,z)$ of $\calG_{n,N}$, we 
  have that $z=\prod_{i\in [n]}x_i$.
\label{util_s5}
\end{lemma}

\begin{proof}
Assume for contradiction that $z>\prod_{i} x_i$. 
As $u_{Q}(q_2) > u_{Q}(q_1)$ and $\calX$ is a Nash equilibrium,
  player $Q$ never plays $q_1$ and thus, $y=0$.  
This in turn implies $u_{R}(r_2)=1 >0= u_{R}(r_1)$ and $z=0$,
which contradicts with the assumption that $z>\prod_i x_i\ge 0$.

Next, assume for contradiction that $z<\prod_i x_i$, giving
  us that $u_{Q}(q_2) < u_{Q}(q_1)$. Player $Q$ never plays $q_2$ and $y=1$. 
This implies that $u_{R}(r_1)>u_{R}(r_2)$ and thus
  $z=1$, which contradicts with the assumption that $z<\prod_i x_i\le 1$ (as $x_i\in [0,1]$).
This finishes the proof of the lemma.
\end{proof}


We now show that the radix game $\calG_{n,N}$ has a unique
  Nash equilibrium $\calX$ with $x_i=\delta^i$. 
\begin{lemma}\label{notuseful2}
In a Nash equilibrium $\calX=(x_1,\ldots,x_n,y,z)$ of $\calG_{n,N}$, 
  we have $x_i = \delta^i$ for all $i\in [n]$.
\label{powers}
\end{lemma}

\begin{proof}
First we show that $\prod_{i\in [n]} x_i=\prod_{i\in [n]} \delta^i$.
Consider for contradiction the following two cases:
\begin{flushleft}\begin{enumerate}\vspace{0.1cm}
\item[] {Case 1}: $\prod_{i\in [n]} x_i<\prod_{i\in [n]} \delta^i$. Then 
  there is an $i\in [n]$ such that $x_i<\delta^i$. 
For $P_i$, we have 
\begin{equation}\label{eq:ha1}
u_i(s)=\delta^i\prod_{j\ne i} x_j + \prod_{j\in [n]} \delta^j
> \prod_{j\in [n]} x_j + \prod_{j\in [n]} x_j=
2\prod_{j\in [n]} x_j=2z=u_i(t).
\end{equation}
This implies that $x_i=1$, contradicting with the assumption 
  that $x_i<\delta^i<1$ as $N\ge 2$.\vspace{-0.06cm}

\item[] {Case 2}: $\prod_{i\in [n]} x_i>\prod_{i\in [n]} \delta^i$. Then
  there is an $i\in [n]$ such that $x_i>\delta^i$.
For $P_i$, we have
\begin{equation}\label{eq:ha2}
u_i(s)=\delta^i \prod_{j\ne i} x_j + \prod_{j\in [n]} \delta^j
< \prod_{j\in [n]} x_j + \prod_{j\in [n]} x_j=
2\prod_{j\in [n]} x_j=2z=u_i(t).
\end{equation}
This implies that $x_i=0$, contradicting with the assumption
  that $x_i>\delta^i>0$.\vspace{0.1cm}
\end{enumerate}\end{flushleft}
As a result, we must have $\prod_i x_i=\prod_i \delta^i$,
  which also implies that $x_i>0$ for all $i\in [n]$.\vspace{0.01cm}
  
Now we show that $x_i=\delta^i$ for all $i$.
Assume for contradiction that $x_i\ne \delta^i$ for some $i\in [n]$.\vspace{0.1cm}
\begin{flushleft}\begin{enumerate}
\item[] {Case 1}: $x_i<\delta^i$. Then 
the same strict inequality (\ref{eq:ha1}) holds for $P_i$, which implies that 
  \\$x_i=1$, contradicting with the assumption that $x_i<\delta^i<1$ as $N\ge 2$. 
\item[] {Case 2}: $x_i>\delta^i$. Then
the same strict inequality (\ref{eq:ha2}) holds for $P_i$, which implies that
  \\$x_i=0$, contradicting with the assumption that $x_i>\delta^i>0$.\vspace{0.1cm}
\end{enumerate}\end{flushleft}
This finishes the proof of the lemma.
\end{proof}

Notice that Lemma \ref{notuseful1} and \ref{notuseful2} together imply
  that $\calG_{n,N}$ has a unique Nash equilibrium because
  of Lemma \ref{notuseful1} as well as the fact that $0<z<1$ implies
  $u_R(r_1)=y=1-y=u_R(r_2)$ and thus $y=1/2$.

\subsection{Generalized Radix Game}

We use $\calG_{n,N}$ to define an anonymous 
  game $\calG^*_{n,N}$, called the \emph{generalized radix game}, with
  the same set of $n+2$ players $\{P_1,\ldots,P_n,Q,R\}$ but seven strategies 
  $\{s_1,s_2,t,q_1,q_2,r_1,r_2\}$.
To this end, we replace strategy $s$ in $\calG_{n,N}$ with two of its
  duplicate strategies $s_1,s_2$ in $\calG^*_{n,N}$ and make sure that 
  the players in $\calG^*_{n,N}$ treat both $s_1$ and $s_2$ the same as 
  the old strategy $s$, and have their payoff 
  functions derived from those of players in $\calG_{n,N}$ in this fashion.
We will show in the next section that in any Nash equilibrium of
  $\calG^*_{n,N}$, player $P_i$ must have probability exactly $\delta^i$
  distributed among $s_1,s_2$.
   
For readers who are familiar with previous PPAD-hardness results of Nash
  equilibria in normal form games \cite{DGPJournal, 2Nash},
  this is the same trick used to derive the game \emph{generalized matching pennies}
  from \emph{matching pennies}.
We define $\calG^*_{n,N}$ formally as follows.

\begin{mygame}[\hspace{0.02cm}Generalized Radix Game $\calG_{n,N}^*$\hspace{0.02cm}]
Let $n\ge 1$ and $N\ge 2$ be two parameters. Let $\delta=1/N$. 
We use $\calG^*_{n,N}$ to denote an anonymous game 
  with the same $n+2$ players $\{P_1,\ldots,P_n,Q,R\}\hspace{-0.03cm}$ as 
  $\calG_{n,N}$ but now $7$ pure strategies $\{s_1,s_2,t,q_1,q_2,r_1,r_2\}$.
The payoff function $\emph{\pay}_T^*$ of a player $T$ in $\calG^*_{n.N}$ is 
  defined using $\emph{\pay}_T$ of the same player $T$ in $\calG_{n,N}$ as follows:
$$
\emph{\pay}_T^*\Big(b,\big(k_{s_1},k_{s_2},k_t, k_{q_1},k_{q_2},k_{r_1},k_{r_2}\big)\Big)=
\emph{\pay}_T\Big(\phi(b),\big(k_{s_1}+k_{s_2},k_t, k_{q_1},k_{q_2},k_{r_1},k_{r_2}\big)\Big),
$$
where $\phi(s_1)=\phi(s_2)=s$ and $\phi(b)=b$ for every other pure strategy.
\end{mygame}

Since the payoff of player $P_i$ is always $-1$ when playing $q_1,q_2,r_1$ or $r_2$,
  she is only interested in $s_1,s_2$ and $t$.
Similarly $Q$ is only interested in $q_1,q_2$ and $R$ is only interested in $r_1,r_2$.
As a result, a Nash equilibrium $\calX$ of $\calG^*_{n,N}$ can be
  fully specified by $2n+2$ numbers $(x_{i,1},x_{i,2},y,z:i\in [n])$,
  where $x_{i,1}$ (or $x_{i,2}$) denotes the probability of
  $P_i$ playing strategy $s_{1}$ (or strategy $s_2$, respectively), so the probability of
  $P_i$ playing $t$ is $1-x_{i,1}-x_{i,2}$.
We also let $x_i=x_{i,1}+x_{i,2}$ for each $i\in [n]$.

Given the definition of $\calG^*_{n,N}$ from $\calG_{n,N}$, 
  Lemma \ref{notuseful2} 
  suggests $x_i=x_{i,1}+x_{i,2}=\delta^i$, for all $i\in[n]$,
  in every Nash equilibrium $\calX$ of $\calG^*_{n,N}$.
This indeed follows from the main lemma of the next section
  concerning $\epsilon$-well-supported Nash equilibria of not only 
  the generalized radix game $\calG_{n,N}^*$ itself, but also 
  anonymous games obtained by
  perturbing payoff functions of $\calG_{n,N}^*$.

\section{Generalized Radix Game after Perturbation}

In this section, we analyze $\epsilon$-well-supported Nash equilibria of 
  anonymous games obtained by perturbing payoff functions of
  the generalized radix game $\calG^*_{n,N}$.
Recall that $n\ge 1$ and $N\ge 2$, and 
  we use $\pay_T^*$ to denote the payoff function of a player $T$ in $\calG^*_{n,N}$.
Given $x,y\in \mathbb{R}$ and $\xi\ge 0$, we write $x=y\pm \xi$ to denote $|x-y|\le \xi$. 
We first define anonymous games that are close to $\calG_{n,N}^*$.

\begin{definition}
For $\xi\ge 0$, we say an anonymous game 
  $\calG$ is \emph{$\xi$-close} to $\calG_{n,N}^*$ if
\begin{flushleft}\begin{enumerate}
\item $\calG$ has the same set $\{P_1,\ldots,P_n,Q,R\}$ of players and same set 
  of $7$ strategies as $\calG_{n,N}^*$.\vspace{-0.08cm}
 
\item For each player $T\in \{P_1,\ldots,P_n,Q,R\}$, her payoff function $\emph{\pay}_T$  
  in $\calG$ satisfies  
$$
\textsf{\emph{payoff}}_T(b,\kk)=\textsf{\emph{payoff}}_T^*(b,\kk)\pm \xi,
$$
for all $b\in \{s_1,s_2,t,q_1,q_2,r_1,r_2\}$ and all 
  histograms $\kk$ of strategies played by $n+1$ players.
\end{enumerate}\end{flushleft}
\end{definition}

To characterize $\epsilon$-well-supported Nash equilibria of a game $\calG$ 
  $\xi$-close to $\calG_{n,N}^*$
we first show that when $\epsilon$, $\xi$ are small enough,
  each player in $\calG$ remains only interested in a subset 
  of strategies, i.e., $\{s_1,s_2,t\}$ for $P_i$, $\{q_1,q_2\}$ for $Q$, and $\{r_1,r_2\}$ for $R$,
  in any $\epsilon$-well-supported Nash equilibrium of $\calG$.

\begin{lemma}\label{lem:interested}
Let $\calG$ be an anonymous game $\xi$-close to $\calG_{n,N}^*$
  for some $\xi\ge 0$.
When $2\hspace{0.02cm}\xi+\epsilon<1$, every $\epsilon$-well-supported Nash equilibrium
  of $\calG$ satisfies: player $P_i$ only plays $\{s_1,s_2,t\}$; player
  $Q$ only plays $\{q_1,q_2\}$; player $R$ only plays $\{r_1,r_2\}$.
\end{lemma}
\begin{proof}
We only prove (1) since the proof of (2) and (3) is similar.

Given an $\epsilon$-well-supported Nash equilibrium $\calX$, as the payoff of $P_i$ when playing
  $b\notin \{s_1,s_2,t\}$ is always $-1$ in $\calG_{n,N}^*$,
  her expected payoff when playing $b$ in $\calG$ is at most $-1+\xi$;
as the payoff of $P_i$ when playing $b\in \{s_1,s_2,t\}$ is always nonnegative
  in $\calG_{n,N}^*$, her expected payoff in $\calG$ is at least $-\xi$.
It follows from $2\hspace{0.02cm}\xi+\epsilon<1$ and the assumption of $\calX$ being an 
  $\epsilon$-well-supported equilibrium that $P_i$ only plays strategies in $\{s_1,s_2,t\}$
  with positive probability.
\end{proof}

It follows from Lemma \ref{lem:interested} that an $\epsilon$-well-supported Nash
  equilibrium of $\calG$ can be fully described by~a tuple of $2n+2$ numbers
  $(x_{i,1},x_{i,2},y,z:i\in [n])$, when $\xi,\epsilon$ satisfy $2\hspace{0.02cm}\xi+\epsilon<1$:
  $x_{i,1}$ denotes the probability of $P_i$ playing $s_1$, $x_{i,2}$ denotes 
  the probability of $P_i$ playing $s_2$, $y$ denotes the probability of $Q$
  playing $q_1$, and $z$ denotes the probability of $R$ playing $r_1$.

Recall that $\delta=1/N\le 1/2$. Let $\kappa=\prod_{i\in [n]} \delta^i.$
We prove the main lemma of this section.

\begin{lemma}\label{lem:perturb}
Let $\calG$ denote an anonymous game that is $\xi$-close 
  to $\calG_{n,N}^*$.
Suppose that $\xi,\epsilon\ge 0$ satisfy
\begin{equation}\label{condition}
\tau=\frac{36\hspace{0.02cm}\xi+18\hspace{0.02cm}\epsilon}{\kappa}\le 1/2.
\end{equation}
Then every $\epsilon$-well-supported Nash equilibrium of $\calG$ satisfies
$x_{i,1}+x_{i,2}=\delta^i\pm \tau\delta^i$ for all $i\in [n]$.
\end{lemma}
\begin{proof}
Let $\hspace{-0.03cm}\calX\hspace{-0.03cm}=(x_{i,1},x_{i,2},y,z:i\in [n])$ be an 
  $\epsilon$-well-supported Nash equilibrium of $\calG$.
For each $i\in $ $[n]$ we let $x_i=x_{i,1}+x_{i,2}$.
Since $\calG$ is $\xi$-close to $\calG_{n,N}^*$, we have the following estimates:
\begin{flushleft}\begin{enumerate}
\item The expected payoff of $P_i$ for playing strategy $s_1$ or $s_2$ is 
\begin{eqnarray*}
&u_i(s_1),\hspace{0.04cm}u_{i}(s_2)=\left(\delta^i\cdot \Pr\big[k_{s_1}+k_{s_2}=n-1\big] + \prod_{j \in [n]} \delta^j\right)\pm\xi 
= \left(\delta^i \prod_{j\ne  i} x_j +\kappa\right) \pm\xi,&\end{eqnarray*}
where we write $k_{s_1},k_{s_2}$ to denote the numbers of players that play $s_1,s_2$ respectively,
  as seen by player $P_i$ (same below). 
The expected payoff of $P_i$ for playing $t$ is $u_{i}(t)=2z\pm \xi$.

\item The expected payoff of $Q$ for playing $q_1$ is 
\begin{eqnarray*}
&u_{Q}(q_1)= \Pr\big[k_{s_1}+k_{s_2}=n\big]\pm \xi=\prod_{j\in [n]} x_j\pm \xi.&
\end{eqnarray*}
The expected payoff of $Q$ for playing $q_2$ is 
$u_Q(q_2)=z\pm \xi$.

\item
The expected payoff of $R$ for playing $r_1$ is $u_R(r_1)=y\pm \xi$
  and for $r_2$ is $u_R(r_2)=(1-y)\pm \xi$.  
\end{enumerate}\end{flushleft}
To rest of the proof follows those of Lemma \ref{util_s5} and Lemma \ref{powers}.
First we show that $z$ must satisfy
\begin{equation}\label{haha1}
 z=\prod_{j\in [n]} x_j\pm (2\xi+\epsilon).
\end{equation}
The proof is the same as that of Lemma \ref{util_s5}, 
  using the assumption that $\calX$ is $\epsilon$-well-supported.

Given (\ref{haha1}), next we show that the $x_i$'s satisfy
\begin{equation}\label{haha2}
\prod_{i\in [n]}x_i=\prod_{i\in [n]} \delta^i\pm 
  (6 \xi+3\epsilon)=\kappa\pm (6 \xi+3\epsilon).
\end{equation}
To this end we follow the proof of the first part of Lemma \ref{powers} and
  consider the following two cases:\vspace{0.06cm}
\begin{flushleft}\begin{enumerate}
\item[] {Case 1}: $\prod_{i\in [n]} x_i<\kappa-(6\xi+3\epsilon)$. Then 
  there exists an $i\in [n]$ such that $x_i<\delta^i$. 
For $P_i$: 
$$
\hspace{-0.1cm}u_i(s_1)\ge \delta^i\prod_{j\ne i} x_j + \kappa-\xi
> 2\prod_{j\in [n]} x_j +5\xi+3\epsilon\ \ \ \text{and}\ \ \ 
u_i(t)\le 2z+\xi\le 2\prod_{j\in [n]}x_j+5\xi+2\epsilon.
$$
This implies that $P_i$ does not play $t$ in $\calX$, an $\epsilon$-well-supported 
  Nash equilibrium of $\calG$, and thus, $x_i=x_{i,1}+x_{i,2}=1$, contradicting with  
  $x_i<\delta^i<1$ as $N\ge 2$. 

\item[] {Case 2}: $\prod_{i\in [n]} x_i>\kappa+(6\xi+3\epsilon)$. Then
  there exists an $i\in [n]$ such that $x_i>\delta^i$.
For $P_i$:
$$
\hspace{-0.1cm}u_i(s_1),\hspace{0.03cm}u_i(s_2)\le \delta^i \prod_{j\ne i} x_j + \kappa+\xi
< 2\prod_{j\in [n]} x_j -5\xi-3\epsilon\ \ \ \text{and}\ \ \ 
u_i(t)\ge 2\prod_{j\in [n]} x_j-5\xi-2\epsilon.
$$
This implies that $P_i$ plays neither $s_1$ nor $s_2$ and thus,
  we have $x_{i,1}=x_{i,2}=0$ and $x_i=0$ as well, contradicting with 
  $x_i>\delta^i>0$.\vspace{0.06cm}
\end{enumerate}\end{flushleft}
By (\ref{haha1}) and (\ref{haha2}), $z=\kappa\pm (8\xi+4\epsilon)$.
(\ref{haha2}) also implies that $x_i>0$ since $\kappa>0$
  and $\kappa\ge 72\xi+36\epsilon$ by (\ref{condition}).

Finally, assume for contradiction that either $x_i<(1-\tau)\delta^i$
  or $x_i>(1+\tau)\delta^i$ for some $i\in [n]$.\vspace{0.08cm}
\begin{flushleft}\begin{enumerate}
\item[] Case 1: {$x_i<(1-\tau)\delta^i$}. Then using $\tau\le 1/2$ and
  $1\le 1/(1-\tau)\le 2$, we have
\begin{align*}
u_i(s_1)-u_i(t)\ge \delta^i\prod_{j\ne i} x_j+
\kappa-2z-2\xi>  \frac{\kappa-6\xi-3\epsilon}{1-\tau} 
 +\kappa -2z-2\xi\ge\tau\kappa-30\xi-14\epsilon.
\end{align*}
Plugging in the definition of $\tau$ in (\ref{condition}), we have
  $u_i(s_1)-u_i(t)>\epsilon$ and thus, $x_i=1$, which contradicts
  with the assumption that $x_i<(1-\tau)\delta^i<1$.
\item[] Case 2: $x_i>(1+\tau)\delta^i$. Then using $\tau\le 1/2$ and
  $2/3\le 1/(1+\tau)\le 1$, we have
\begin{align*}
u_i(s_1)-u_i(t)\le \delta^i\prod_{j\ne i} 
  x_j+\kappa-2z+2\xi < \frac{\kappa+6\xi+3\epsilon}{1+\tau}
  +\kappa -2z+2\xi\le -\frac{2\tau\kappa}{3}+24\xi+11\epsilon.
\end{align*}
The same inequality holds for $u_i(s_2)-u_i(t)$.
Plugging in (\ref{condition}), we have $u_i(s_1)-u_i(t)<-\epsilon$
  as well as $u_i(s_2)-u_i(t)<-\epsilon$.
This in turn implies that $x_{i,1}=x_{i,2}=0$ and thus, $x_i=0$, which 
  contradicts with the assumption that $x_i>(1+\tau)\delta^i>0$.\vspace{0.08cm}
\end{enumerate}\end{flushleft}
This finishes the proof of the lemma.
\end{proof}

\section{Reduction from Polymatrix Games to Anonymous Games}\label{main reduction}

In this section we prove the hardness part of Theorem \ref{main-theorem}.
For this purpose we present a polynomial time reduction from the problem
  of finding a $1/n$-well-supported Nash equilibrium in a polymatrix game
  to the problem of finding an $\epsilon$-well-supported Nash equilibrium
  in an anonymous game with $7$ strategies, for some exponentially small $\epsilon$.
We first give some intuition behind this quite involved reduction
  in Section \ref{hahasec}.
Details of the reduction and the proof of its correctness 
  are then presented in Section \ref{sec:reduction} and \ref{sec:correctness}, respectively,
  with a key technical lemma proved in Section \ref{sec:estimation}.
We finish the proof of the hardness part  
  in Section \ref{sec:hardness} by showing that any approximate Nash equilibrium of~an anonymous game
  can be converted into a well-supported equilibrium efficiently
  (since Theorem \ref{main-theorem} is concerned with approximate Nash equilibria).
  


\subsection{Overview of the Reduction}\label{hahasec}

Given as input a polymatrix game specified by a
  matrix $\AA \in [0,1]^{2n\times 2n}$, our goal is to construct~in polynomial
  time an anonymous game $\calG_\AA$, and 
  show that every $\epsilon$-well-supported Nash equilibrium
  of $\calG_\AA$, where $\epsilon=1/2^{n^6}$, can be used to recover 
  a $(1/n)$-well-supported equilibrium of $\AA$ in polynomial time.
Note that this is not exactly the PPAD-hardness result as claimed 
  in Theorem \ref{main-theorem} but we will fill in the gap
  in Section \ref{sec:hardness} with some standard arguments.

Given $\AA$, we construct $\calG_\AA$ by perturbing payoff functions of the 
  Generalized Radix game ${\cal{G}}^*_{n,N}$ with $N=2^n$,
  so that $\calG_\AA$ is $\xi$-close to $\calG^*_{n,N}$ for some 
  exponentially small $\xi>0$ to be specified later.
(Thus, $\calG_\AA$ has the same set of $n+2$ players 
  $\{P_1,\ldots,P_n,Q,R\}$ as well as the same set of $7$ strategies  
  $\{s_1,s_2,t,q_1,q_2,r_1,r_2\}$  as $\calG^*_{n,N}$.) 
By Lemma \ref{lem:interested} and Lemma \ref{lem:perturb} we know that
  every $\epsilon$-well-supported equilibrium of $\calG_\AA$
  can be fully described by a tuple $\calX=(x_{i,1},x_{i,2},y,z:i\in [n])$
  that satisfies
\begin{equation}\label{eq:approximate}
x_{i,1}+x_{i,2}\approx \delta^i
\end{equation}
  for each $i\in [n]$, where $\delta=1/N=1/2^n$.

Our construction of $\calG_\AA$  
  has player $P_\ell$ simulate 
  row $2\ell-1$ and $2\ell$ of the polymatrix game $\AA$ for each $\ell\in [n]$. 
The goal is to show at the end that,
  after normalizing $(x_{\ell,1},x_{\ell,2})$, i.e., probabilities of $P_\ell$ playing $s_1,s_2$
  in an $\epsilon$-well-supported equilibrium $\calX$ of $\calG_\AA$,
  into a distribution $(y_{2\ell-1},y_{2\ell})$:
\begin{equation}\label{eq:back}
y_{2\ell-1}=\frac{x_{\ell,1}}{x_{\ell,1}+x_{\ell,2}}\ \ \ \text{and}\ \ \ 
y_{2\ell}=\frac{x_{\ell,2}}{x_{\ell,1}+x_{\ell,2}},
\end{equation}
we get a $(1/n)$-well-supported Nash equilibrium $\yy=(y_1,\ldots,y_{2n})$ of $\AA$.
By (\ref{eq:approximate}) we have
$$
y_{2\ell-1}\approx N^\ell\cdot x_{\ell,1}\ \ \ \text{and}\ \ \ 
y_{2\ell}\approx N^{\ell}\cdot x_{\ell,2}.
$$

For player $P_\ell$ to simulate row $2\ell-1$ and $2\ell$ of the polymatrix game $\AA$,
  we perturb the original payoff function $\pay^*_\ell$ of  
  $P_\ell$ in ${\cal{G}}^*_{n,N}$ in a way such that the following two linear forms of $\yy$:
$$
\AA_{2\ell-1}\cdot \yy=\sum_{j\notin \{2\ell-1,2\ell\}}A_{2\ell-1,j}\cdot y_j\ \ \ \text{and}
\ \ \ \AA_{2\ell}\cdot \yy=\sum_{j\notin\{2\ell-1,2\ell\}} A_{2\ell,j}\cdot y_j
$$  
appear as additive terms in the expected payoffs $u_\ell(s_1,\calX)$ and 
  $u_\ell(s_2,\calX)$ of $P_\ell$ obtained from $s_1,s_2$, respectively.
Let $u^*_\ell(\sigma, {\bxx})$ denote the expected payoff of player $P_\ell$ in 
  the original generalized radix game ${\cal{G}}^*_{n,N}$ for strategies $\sigma \in \{s_1,s_2\}$.
Then more specifically, 
  we would like to perturb carefully~the payoff functions of ${\cal{G}}^*_{n,N}$ such that 
  for every $\ell\in [n]$, the expected payoffs of player $P_\ell$ in 
  an $\epsilon$-well-supported Nash equilibrium $\bxx$ of $\cal{G}_\AA$ satisfy\vspace{0.06cm}
\begin{align} 
u_\ell(s_1,\bxx)&\approx u^*_\ell(s_1,\bxx) + \xi^*\cdot \AA_{2\ell-1}\cdot \yy  \nonumber\\ \label{reduct}
&\approx u^*_\ell(s_1,\bxx)+\xi^* \sum_{j\ne \ell} N^j
\big(A_{2\ell-1,2j-1}\cdot x_{j,1}+A_{2\ell-1,2j}\cdot x_{j,2}\big)\\[0.5ex]
 u_\ell(s_2, {\bxx})&\approx 
u^*_\ell(s_2, {\bxx}) + 
\xi^* \cdot \AA_{2\ell}\cdot \yy \nonumber\\ \label{reduct2}
&\approx u^*_\ell(s_2,\bxx)+\xi^* \sum_{j\ne \ell} N^j
\big(A_{2\ell,2j-1}\cdot x_{j,1}+A_{2\ell,2j}\cdot x_{j,2}\big)
\end{align}
where $\xi^*$ is a parameter small enough to make sure that 
  the resulting game is $\xi$-close to $\calG_{n,N}^*$.

If one can perturb the payoff functions of players $P_\ell$ in $\calG^*_{n,N}$ 
  so that (\ref{reduct}) and (\ref{reduct2}) hold for every $\epsilon$-well-supported 
  Nash equilibrium $\calX$ of $\calG_\AA$, then the vector $\yy$ obtained
  from $\calX$ using (\ref{eq:back}) must be a $(1/n)$-well-supported equilibrium of $\AA$. 
To see this, assume for contradiction that 
\begin{equation}\label{eq:contradiction}
\AA_{2\ell-1}\cdot \yy>\AA_{2\ell}\cdot \yy+1/n
\end{equation}
but $y_{2\ell}>0$. Using (\ref{eq:contradiction}), (\ref{reduct}), and (\ref{reduct2}), we have
  $u_\ell(s_1,\bxx)$ is bigger than $u_\ell(s_2,\bxx)$ by $\xi^*/n$
  (assuming that errors hidden in both (\ref{reduct}) and (\ref{reduct2}) are negligible). 
As long as our choice of $\xi^*$ satisfies $\xi^*/n>\epsilon$ 
  we must have $x_{\ell,2}=0$ and thus, $y_{2\ell}=0$ from (\ref{eq:back}).
  
However, perturbing the generalized radix game so that (\ref{reduct}) and (\ref{reduct2})
  hold is challenging. While\vspace{0.03cm}
\begin{equation}\label{eq:desired}
\sum_{j\ne \ell} N^j\big(A_{2\ell-1,2j-1}\cdot x_{j,1}+A_{2\ell-1,2j}\cdot x_{j,2}\big)\ \ \ 
\text{and}\ \ \ 
\sum_{j\ne \ell} N^j\big(A_{2\ell,2j-1}\cdot x_{j,1}+A_{2\ell,2j}\cdot x_{j,2}\big)
\end{equation}
are merely two linear forms of $(x_{j,1},x_{j,2}:j\ne \ell)$ from $\calX$, they are
  extremely difficult to obtain due to the nature of anonymous games: 
  the expected payoff of player $P_\ell$ is
\begin{equation}\label{eq:expected}
u_\ell(\sigma,\calX)=\sum_{\kk \in K} \pay_\ell(\sigma,\kk) \cdot \Pr_{\cal{X}}[P_\ell,\kk] ,
\end{equation}
a linear form of $\Pr_{\cal{X}}[P_\ell,\kk]$, the probability of $P_\ell$ 
  seeing histogram $\kk$ given $\calX$.
As each $\Pr_{\cal{X}}[P_\ell,\kk]$ is a highly complex and symmetric expression of 
  variables in $\calX$, it is not clear how one can extract from (\ref{eq:expected})
  the desired linear forms of (\ref{eq:desired}).

This is where the fact that $x_{i,1}+x_{i,2}\approx \delta^i$ 
  helps us tremendously. (Recall that this holds as long as the generalized
  radix game $\calG^*_{n,N}$ and  $\calG_\AA$ are $\xi$-close.)
The core of the construction of~$\calG_\AA$ uses the following key technical lemma which
  we refer to as the \emph{estimation lemma}.
It shows that under any mixed 
  strategy profile $\calX=(x_{i,1},x_{i,2},y,z:i\in [n])$ such that $x_{i,1}+x_{i,2}\approx \delta^i$,
  there is indeed a linear form of $\Pr_{\cal{X}}[P_\ell,\kk]$ that gives us a close 
  approximation of $x_{j,1}$ (or $x_{j,2}$), $j\ne \ell$, and its coefficients
  can be computed in polynomial time in $n$.
We delay its proof to Section \ref{sec:estimation}.

\begin{lemma}[Estimation Lemma] \label{lem:estimation-lemma}
Let $N=2^n$ and ${\lambda= 2^{-n^3}}$.
Given $\ell\in [n]$ and $j\ne \ell\in [n]$ one can compute 
  in polynomial time in $n$ vectors $\BB^{[\ell,j]},\CC^{[\ell,j]}$
  of length $|K|$ \emph{(}indexed by $\kk\in K$\emph{)} such that every
  mixed strategy profile $\calX=(x_{i,1},x_{i,2},y,z:i\in [n])$ with
  $x_{i,1}+x_{i,2}=\delta^i\pm \lambda$ for all $i$ satisfies
\begin{equation*}
\sum_{\kk \in K} B^{[\ell,j]}_\kk \cdot \emph{\Pr}_{\calX}[P_\ell,\kk] = x_{j,1} \pm 
  O\hspace{-0.05cm}\left( j^2 \delta^{j+1}\right)\ \ \ \text{and}\ \ \ 
\sum_{\kk \in K} C^{[\ell,j]}_\kk \cdot \emph{\Pr}_{\calX}[P_\ell,\kk] = x_{j,2} \pm 
  O\hspace{-0.05cm}\left( j^2 \delta^{j+1}\right).
\end{equation*} 
Moreover, the absolute value of each entry of $\BB^{[\ell,j]}$ and 
  $\CC^{[\ell,j]}$ is at most $N^{n^2}$. 
\end{lemma}

With the estimation lemma in hand we can derive linear forms of 
  $\Pr_{\cal{X}}[P_\ell,\kk]$ that are close approximations of the two linear forms
  of $(x_{j,1},x_{j,2}:j\ne \ell)$
  in (\ref{eq:desired}). We then use the coefficients of these
  linear forms of $\Pr_{\cal{X}}[P_\ell,\kk]$ to perturb $\calG_{n,N}^*$ and wrap up the construction of $\calG_\AA$.

\subsection{Construction of Anonymous Game $\calG_\AA$}\label{sec:reduction}

Let $\AA\in [0,1]^{2n\times 2n}$ denote the input polymatrix game.
We need the following parameters:
$$
N=2^n,\ \ \delta={1}/{N}= {2^{-n}},\ \ \lambda={2^{-n^3}},\ \
\xi={2^{-n^4}},\ \ \xi^*={2^{-n^5}}\ \ 
\text{and}\ \ \epsilon={2^{-n^6}}.
$$
We remark that we do not attempt to optimize the parameters here
  but rather set them in different scales to facilitate the analysis later.

\begin{mygame}[Construction of $\calG_\AA$]
We use the polynomial-time algorithm promised in the Estimation Lemma to
  compute $\BB^{[\ell,j]}$ and $\CC^{[\ell,j]}$, for all $\ell\in [n]$ and $j\ne \ell\in [n]$.

Starting with the generalized radix game ${\cal{G}}^*_{n,N}$,
  we modify payoff functions of players $P_1,\ldots,P_n$ as 
  follows \emph{(}payoff functions of $Q$ and $R$ remain unchanged\emph{)}.
Let $\emph{\pay}^*_\ell$ denote the payoff function of $P_\ell$ in $\calG_{n,N}^*$.
Then for each player $P_{\ell}$ and each histogram $\kk \in K$, we set 
\begin{align*}
\emph{\pay}_{{\ell}} (s_1,\kk) &= \emph{\pay}_{{\ell}}^* (s_1,\kk) + 
  \xi^*\sum_{j\ne \ell} N^j\left(A_{2\ell-1,2j-1}\cdot B_\kk^{[\ell,j]}+
  A_{2\ell-1,2j}\cdot C_{\kk}^{[\ell,j]}\right)\\[0.5ex]
\emph{\pay}_{\ell} (s_2,\kk) &= \emph{\pay}^*_{{\ell}} (s_2,\kk) + \xi^*
  \sum_{j\neq \ell} N^j\left(A_{2\ell, 2j - 1}\cdot B_\kk^{[\ell,j]} + 
    A_{2\ell, 2j}\cdot C_\kk^{[\ell,j]}\right), 
\end{align*} 
and keep all other payoffs of $P_\ell$ the same
  \emph{(}i.e., $\emph{\pay}_\ell(\sigma,\kk)=
  \emph{\pay}_\ell^*(\sigma,\kk)$ for all
  $\sigma\notin\{s_1,s_2\}$\emph{)}.
\end{mygame}

A few properties of $\calG_\AA$ then follow directly from its construction.
First, observe that entries of $\AA$ lie in $[0,1]$ and entries of
  $\BB^{[\ell,j]}$ and $\CC^{[\ell,j]}$ have absolute values at most $N^{n^2}=2^{n^3}$.
We have

\begin{property}
Given $\AA\in [0,1]^{2n\times 2n}$, $\cal{G}_\AA$ is an 
  anonymous game $\xi$-close to 
  ${\cal{G}}^*_{n,N}$ where $\xi ={2^{-n^4}}$.
\end{property}

By Lemma \ref{lem:interested}, an $\epsilon$-well-supported Nash
  equilibrium of $\calG_\AA$ is fully described by
  a $(2n+2)$-tuple $\calX=(x_{i,1},x_{i,2},y,z:i\in [n])$,
  where  $P_i$ plays strategies $s_1,s_2$ and $t$ with probabilities
  $x_{i,1},x_{i,2}$ and $1-x_{i,1}-x_{i,2}$, respectively.
We also get the following corollary from Lemma \ref{lem:perturb}.

\begin{corollary}
Every $\epsilon$-well-supported equilibrium $\calX
  =(x_{i,1},x_{i,2},y,z:i\in [n])$ of $\calG_\AA$ satisfies
$$
x_{i,1}+x_{i,2}=\delta^i\pm \lambda,\ \ \ \text{for all $i\in [n]$.}
$$
\end{corollary}

Therefore, the conditions of the estimation lemma are met.
It follows that

\begin{property}\label{simpleproperty}
Given an $\epsilon$-well-supported equilibrium $\calX$ of
  $\calG_\AA$, the expected payoffs of $P_\ell$ satisfy\vspace{0.12cm}
\begin{align*}
u_\ell(s_1,\calX)&=u_\ell^*(s_1,\calX)+\xi^*\sum_{j\ne \ell}
  N^j\big(A_{2\ell-1,2j-1}\cdot x_{j,1} +A_{2\ell-1,2j}\cdot x_{j,2}\big)\pm 
  O(n^3\xi^*\delta)\ \ \ \ \text{and}\\[0.6ex]
u_\ell(s_2,\calX)&=u_\ell^*(s_2,\calX)+\xi^*\sum_{j\ne \ell}
  N^j\big(A_{2\ell,2j-1}\cdot x_{j,1} +A_{2\ell,2j}\cdot x_{j,2}\big)\pm 
  O(n^3\xi^*\delta).
\end{align*}
\end{property}

\subsection{Correctness of the Reduction}\label{sec:correctness}

We are now ready to show that, given an $\epsilon$-well-supported
  Nash equilibrium $\calX$ of $\calG_\AA$,
  the vector $\yy$ derived from $\calX$ using (\ref{eq:back}) 
  is a $(1/n)$-well-supported Nash equilibrium
  of the polymatrix game $\AA$.

\begin{lemma}\label{reduct_1}
Let $\calX=(x_{i,1}, x_{i,2},y,z:i\in [n])$ be an $\epsilon$-well supported 
  Nash equilibrium of ${\cal{G}}_\AA$. 
Then the vector $\yy\in [0,1]^{2n}$ derived from $\calX$
  using \emph{(\ref{eq:back})} is a $(1/n)$-well-supported Nash equilibrium of $\AA$.
\end{lemma}
\begin{proof}
Firstly, note that $x_{i,1}+x_{i,2}>0$ so $\yy$ 
  is well defined and satisfies $y_{2i-1}+y_{2i}=1$ for all $i$.

Assume towards a contradiction that  
  $\yy$ derived from $\calX$ using (\ref{eq:back})
  is not a $(1/n)$-well-supported Nash equilibrium of $\AA$, 
  i.e., there is a player $\ell\in [n]$ such that, without loss
  of generality, 
\begin{equation}\label{final11}
\AA_{2\ell -1}\cdot \yy > \AA_{2\ell}\cdot \yy + 1/n
\end{equation}
but $y_{2\ell}>0$, which in turn implies that $x_{\ell,2}>0$. 

Since $x_{j,1}+x_{j,2}=\delta^j\pm \lambda$, we have
$$
y_{2j-1}=\frac{ x_{j,1}}{x_{j,1}+x_{j,2}} = N^jx_{j,1} 
  \pm\hspace{0.04cm} O(N^{2j}\lambda)=N^jx_{j,1}\pm O\hspace{0.04cm}
  (N^{2n}\lambda).
$$
Similarly we also have $y_{2j}=N^jx_{j,2}\pm O\hspace{0.04cm}(N^{2n}\lambda)$.
Combining these with Property \ref{simpleproperty}, we have\vspace{0.1cm}
\begin{align}\nonumber
u_\ell(s_1,\calX)&=u_\ell^*(s_1,\calX)+\xi^*\cdot \AA_{2\ell-1}\cdot \yy\pm 
  \left(O(n^3\xi^*\delta)+O(n\xi^*N^{2n}\lambda)\right)\ \ \ \ \text{and}\\[0.8ex]
u_\ell(s_2,\calX)&=u_\ell^*(s_2,\calX)+\xi^*\cdot \AA_{2\ell}\cdot \yy\pm 
  \left(O(n^3\xi^*\delta)+O(n\xi^*N^{2n}\lambda)\right).\label{final22} \\[-2.6ex]
  \nonumber
\end{align}
By our choices of parameters, $n\xi^*N^{2n}\lambda\ll n^3\xi^*\delta$
  so the former can be absorbed into the latter.

Combining (\ref{final11}) and (\ref{final22})
  (as well as the fact that $u_\ell^*(s_1,\calX)=u_\ell^*(s_2,\calX)$
  because the payoffs of $s_1$ and $s_2$ are exactly the same 
  in the generalized radix game $\calG^*_{n,N}$), we have
$$
u_\ell(s_1,\calX)-u_\ell(s_2,\calX)\ge \xi^*\left(\AA_{2\ell-1}\cdot \yy-
  \AA_{2\ell}\cdot \yy\right)-O\hspace{0.04cm}(n^3\xi^*\delta)\ge 
  \xi^*/n-O\hspace{0.04cm}(n^3\xi^*\delta)>\epsilon,
$$
for sufficiently large $n$, by our choices of parameters $\delta$, $\xi^*$ and $\epsilon$.
It then follows that $x_{\ell,2}=0$, since $\calX$ is assumed 
  to be an $\epsilon$-well-supported Nash equilibrium of $\calG_\AA$,
  contradicting with $y_{2\ell}>0$.
This finishes the proof.
\end{proof}

\subsection{Proof of the Hardness Part of Theorem \ref{main-theorem}}\label{sec:hardness}

From our definitions of $\calG_{n,N}^*$ and $\calG_\AA$,
  it is clear that all payoffs of $\calG_\AA$ are in $[-1,3]$. 
Using standard arguments (invariance of Nash equilibria under shifting
  and scaling), we can easily see that given an anonymous 
  game $\calG = (n,\alpha,\{\textsf{payoff}_p\})$ such that~all 
  payoffs are in the interval $[a,b]$, where $a,b \in \mathbb{R}$ and $a < b$,
a mixed strategy profile $\bxx$ is an $(b-a)\epsilon$-well-supported  
  equilibrium of $\calG$ if and only if $\bxx$ is an $\epsilon$-well-supported 
  equilibrium of $\calG' = (n,\alpha,\{{\pay}'_p\})$, where 
\begin{equation}\label{shifting_scaling}
\pay'_p(\sigma,\kk) = \frac{\pay_p(\sigma,\kk) - a}{b-a}.
\end{equation}
The new game $\calG'$ now has all payoffs from in $[0,1]$.

As a result, we can construct $\calG'_\AA$ from $\calG_\AA$ 
  in polynomial time such that all payoffs of $\calG_\AA'$ lie in $[0,1]$, 
  and Lemma \ref{reduct_1} holds for all $(\epsilon/4)$-well-supported 
  Nash equilibria of $\calG'_\AA$. 
It follows that

\begin{corollary}\label{before_padding}
Fix any $\alpha\ge 7$. The problem of finding a $2^{-(n^6+2)}$-well-supported Nash
  equilibrium of an anonymous game with $\alpha$ actions and $[0,1]$ payoffs is {PPAD-hard}.
\end{corollary}

This can be further strengthened using a standard padding argument.

\begin{lemma}\label{padding}
Fix any $\alpha\in \mathbb{N}$ and $a>b>0$.
There is a polynomial-time reduction from the problem of finding 
  a $(2^{-n^a})$-well-supported equilibrium to that of finding 
  a $(2^{-n^b})$-well-supported equilibrium, in an anonymous game
  with $\alpha$ actions and $[0,1]$ payoffs.
\end{lemma}

\def\padG{\textsf{pad}\calG}

\begin{proof}
For convenience, we will refer to the problem of finding a 
  $(2^{-n^a})$-well-supported equilibrium as problem $\textsc{A}$ 
  and the other as problem $\textsc{B}$.

Let $\calG\hspace{-0.01cm}=\hspace{-0.01cm}(n,\alpha,\{\pay_p\})$ denote an
  input anonymous game of problem \textsc{A}.
We define a new game $\padG=(n^{t},\alpha,\{\pay_p'\})$ as follows, 
  where $t=a/b>1$ and thus, $n^t>n$. 
To this end, 
  define a map $\phi: \mathbb{Z}^{\alpha} \rightarrow \mathbb{Z}^{\alpha}$ such that 
  $\phi(k_1,\ldots,k_{\alpha}) = (k_1 - (n^t - n), k_2,\ldots,k_{\alpha})$. 
We then define payoff functions of players $\{1,\ldots,n^t\}$ in $\padG$ as follows:
\begin{flushleft}\begin{itemize}
\item For each $i > n$, the payoff function of player $i$ is given by
$$
\text{{\textsf{payoff}}}'_{i}(\sigma,\kk)=\begin{cases}
    \hspace{0.06cm}1& \text{if $\sigma = 1$}\\[0.4ex]
    \hspace{0.06cm}0 & \text{otherwise}
\end{cases}
$$
So player $i$ always plays strategy $1$ 
  in any $\epsilon$-well-supported equilibrium with $\epsilon<1$.
  
\item The payoff of each player $i\in [n]$ is given by 
$$
\text{{\textsf{payoff}}}'_{i}(\sigma,\kk)= 
\begin{cases}
    \hspace{0.06cm} \text{{\textsf{payoff}}}_{i}(\sigma,\phi(\kk))&  
    \text{if $k_1\ge n^t-n$}\\[0.6ex]
    \hspace{0.06cm} 0 & \text{otherwise}
\end{cases}\vspace{0.04cm}
$$
Note that in any $\epsilon$-well-supported equilibrium with $\epsilon<1$,
  the latter case never occurs.
\end{itemize}
\end{flushleft}

By the definition of $\padG$, it is easy to show that 
  $\calX$ is an $\epsilon$-well-supported equilibrium in $\padG$, for some $\epsilon<1$,
  iff 1) each player $i>n$ plays strategy $1$ with probability $1$ and 2)
  the mixed strategy profile of the first $n$ players in $\calX$ is an $\epsilon$-well-supported
  equilibrium of $\calG$.
As a result, a solution~to $\padG$ as an input of problem \textsc{B}
  must be an $\epsilon$-approximate equilibrium of $\calG$ with
  $\epsilon=2^{-(n^t)^b}=2^{-n^a}.$ 
As $\padG$ can be constructed from $\calG$ in polynomial time, 
  this finishes the proof of the lemma.
\end{proof}

Combining Corollary \ref{before_padding} and Lemma \ref{padding}, we have

\begin{corollary}\label{after_padding}
Fix any $\alpha\ge 7$ and $c>0$. The problem of finding a $(2^{-n^c})$-well-supported
  Nash~equilibrium in an anonymous game with $\alpha$ actions and $[0,1]$ payoffs is PPAD-hard.
\end{corollary}

To prove the hardness part of Theorem \ref{main-theorem}, we next give
  a polynomial-time algorithm to compute a well-supported equilibrium 
  from an approximate equilibrium.

\begin{lemma}[From Approximate to Well-Supported Nash Equilibria]\label{poly:equiv}
Let $\calG\hspace{-0.01cm} =\hspace{-0.01cm} (n,\alpha,
  \{\textsf{\emph{payoff}}_p\}\hspace{-0.01cm})$ be an anonymous game with 
  payoffs from $[0,1]$. 
Given an $\epsilon^2/(16\alpha n)$-approximate Nash equilibrium $\bxx$ of $\calG$, 
  one can compute in polynomial time an $\epsilon$-well-supported Nash equilibrium $\byy$ of $\calG$. 
\end{lemma}

\begin{proof}
Let $\bxx=(\xx_i:i\in [n])$ be an $\epsilon'$-approximate Nash equilibrium of $\calG$, with
  $\epsilon'=\epsilon^2/(16\alpha n)$.
For each player $i \in [n]$, we have for any mixed strategy $\xx_i'$, 
\begin{equation}\label{nash_def}
u_i(\xx_i',\bxx_{-i}) \leq u_i(\bxx) +\epsilon',
\end{equation}
where we let $u_i(\xx_i',\bxx_{-i})$ denote the expected payoff of 
  player $i$ when she plays $\xx_i'$ and other players play $\bxx_{-i}$.
Let $\sigma_i$ be a strategy with the highest expected payoff for player $i$ 
  (with respect to $\calX_{-i}$): 
$$u_i(\sigma_i,\bxx)=\max_{k \in[ \alpha]} \hspace{0.06cm}u_i(k,\bxx),$$ and let 
$J_{i} = \{j:u_i(\sigma_i,\bxx) \geq u_i(j,\bxx) + \epsilon/2\}$.
We then define a mixed strategy $\yy_i$ for player $i$ using $\xx_i$, $\sigma_i$ and 
  $J_i$ as follows:
Set $y_{i,j}=0$ for all $j \in J_{i}$, and set 
$$y_{i,\sigma_i}=x_{i,\sigma_i} + \sum_{j \in J_{i}} x_{i,j}.$$ 
All other entries of $\yy_i$ are the same as $\xx_i$. 
As $\yy_i$ 
  increases the expected payoff of player $i$ by at least 
$$({\epsilon}/{2})\cdot \sum_{j \in J_{i}} x_{i,j},$$ 
we have from (\ref{nash_def}) that $\sum_{j \in J_{i}} x_{i,j} \leq 2\epsilon'/\epsilon$.

Repeating this for every player $i\in [n]$, we obtain a new mixed
  strategy profile $\byy$ (clearly $\byy$ can be computed in polynomial time given $\bxx$). 
We finish the proof of the lemma by showing that $\byy$ is 
  indeed an $\epsilon$-well-supported Nash equilibrium of $\calG$.
Below we write $\zeta=2\epsilon'/\epsilon$.
  
First, by the definition of $\byy$, 
  $|x_{i,j} - y_{i,j}| \leq \zeta$ for all $i,j$.
Thus, for any pure strategy profile $\ss_{-i}$,
\begin{align*}
&\prod_{q \neq i} y_{q,s_q} \ge 
\prod_{q \neq i} \max\big\{0, x_{q,s_q} - \zeta\big\} \ge \prod_{q \neq i} x_{q,s_q}
-\zeta \cdot \sum_{q\ne i} \prod_{p\notin \{i,q\}} x_{p,s_p}\ \ \ \ \text{and}\\[0.5ex]
&\prod_{q \neq i} x_{q,s_q} \ge 
\prod_{q \neq i} \max\big\{0, y_{q,s_q} - \zeta \big\} \ge \prod_{q \neq i} y_{q,s_q}
-\zeta \cdot \sum_{q\ne i} \prod_{p\notin \{i,q\}} y_{p,s_p}.
\end{align*}
Since all payoffs are in $[0,1]$, 
  we have for any player $i\in [n]$ and pure strategy $j\in [\alpha]$ that 
\begin{align*}\big|u_i(j,\byy) - u_i(j,\bxx)\big| &\leq \sum_{\ss_{-i}\in [\alpha]^{n-1}}\Bigg|
  \hspace{0.04cm}\prod_{q \neq i} y_{q,s_q}-\prod_{q \neq i} x_{q,s_q}\hspace{0.02cm}\Bigg| \\[1ex]
&\le \zeta \sum_{\ss_{-i}} \hspace{0.04cm}\sum_{q\ne i} \prod_{p\notin \{i,q\}} x_{p,s_p}
+\zeta \sum_{\ss_{-i}} \hspace{0.04cm}\sum_{q\ne i} \prod_{p\notin \{i,q\}} y_{p,s_p}
\\[1ex]&= \zeta \sum_{q\ne i} \hspace{0.04cm}\sum_{\ss_{-i}} \prod_{p\notin \{i,q\}} x_{p,s_p}
+\zeta \sum_{q\ne i} \hspace{0.04cm}\sum_{\ss_{-i}} \prod_{p\notin \{i,q\}} y_{p,s_p}
 \\[0.8ex]
&\le \alpha\hspace{0.02cm}\zeta \sum_{q \neq i} \left(
\sum_{s_r:r\notin \{i,q\}} \hspace{0.04cm}\prod_{p \notin \{i,q\} } x_{p,s_p} \right) + 
\alpha\hspace{0.02cm}\zeta \sum_{q \neq i} \left(\sum_{s_r:r\notin \{i,q\}} 
\hspace{0.04cm}\prod_{p \notin \{i,q\} } y_{p,s_p} \right) \\[0.7ex]
&=2(n-1)\alpha\hspace{0.02cm}\zeta.
\end{align*}
This implies that for any pure strategies $j,k \in [\alpha]$ we have 
$$\big|(u_i(j,\bxx) - u_i(k,\bxx)) - (u_i(j,\byy) - u_i(k,\byy))\big|< \epsilon/2.$$
Therefore, the new mixed strategy profile $\byy=(\yy_i:i\in [n])$ satisfies
$$u_i(j,\byy) < u_i(k,\byy) + \epsilon\ \Rightarrow \ u_i(j,\bxx) < u_i(k,\bxx) + \epsilon/2 
\ \Rightarrow\  y_{i,j} = 0$$ for all $i,j$ and $k$. This finishes the proof of the lemma.
\end{proof}

Fix any $\alpha\ge 7$ and $c>0$.
It then follows from Lemma \ref{poly:equiv} that
  the problem of finding a $(2^{-n^{c/2}})$ well-supported 
  equilibrium in an anonymous game with $\alpha$ actions and $[0,1]$ payoffs is
  polynomial-time reducible to problem $(\alpha,c)$-\anonymous.
As the former problem is PPAD-hard by Corollary \ref{after_padding},
  $(\alpha,c)$-\anonymous\ is PPAD-hard.
The finishes the proof of the hardness part of Theorem \ref{main-theorem}.

\section{Proof of the Estimation Lemma}\label{sec:estimation}

We prove the estimation lemma (Lemma \ref{lem:estimation-lemma}) in this section.

Recall that there are $n$ main players $P_1,\ldots,P_n$,
  and they are only interested in three strategies $\{s_1,s_2,t\}$.  
For convenience we will refer to $s_1$ as strategy $1$, 
  $s_2$ as strategy $2$, and $t$ as strategy $3$ in this section.
Player $P_i$ plays strategy $b\in [3]$ with probability $x_{i,b}$,
  and $\sum_{b} x_{i,b}=1$.
While $x_{i,b}$'s are unknown variables, 
  by the assumption of the lemma we are guaranteed that
\begin{equation}\label{eq:setup}
x_{i,1}+x_{i,2}=\delta^i\pm \lambda,\ \ \ \text{where 
  $\lambda=\delta^{n^2}$.} 
\end{equation}

Throughout this section we will \emph{fix} two distinct integers $r,\ell\in [n]$,
  and the goal will be to derive an approximation of the unknown $x_{r,1}$ for $P_\ell$ 
  using a linear form of the following probabilities:
\begin{equation}\label{linearform}
\Big\{\Pr\big[k_{1} = i, k_{2} = j\big]:i,j \in [0:n-1]\Big\},\ \text{ where $\Pr\big[k_{1} = i, k_{2} = j\big] = \sum_{\substack{\kk \in K \\ k_1 = i, k_2 = j}} \Pr_{\bxx}[P_\ell,\kk]$},
\end{equation}
and $k_{b}$ denotes the random variable that counts 
  players playing $b\in [3]$ other than player 
  $P_\ell$ herself.
We will show that coefficients in 
  the desired linear form can be computed in polynomial time in $n$.  

First we would like to give the reader some intuition for the rest of the section, by showing how one can get a good estimate of $x_{1,1}$ and $x_{2,1}$, assuming $\ell > 2$. We believe this to be useful for more easily understanding the rest of the section, but the reader should feel free to skip it, if desired.

\begin{remark2}[Informal]
As $N=2^n$ is large, we have $x_{i,3} \approx 1$ for each $i$. 
This gives
$$\emph{\Pr}\big[k_{1} = 1, k_{2} = 0\big] \approx x_{1,1} + x_{2,1} + \cdots + x_{n,1}
  = x_{1,1}\pm O(\delta^2)$$ 
as $x_{i,1} \leq \delta^i+\lambda$.
Similarly, $\emph{\Pr}\big[k_{1} = 2, k_{2} = 0\big] 
  \approx x_{1,1}x_{2,1} \pm O(\delta^4)$.
Using $x_{i,1}+x_{i,2}\approx \delta^i$, we have
$$\emph{\Pr}\big[k_{1} = k_{2} = 1\big] \approx x_{1,1}(\delta^2 - x_{2,1}) + x_{2,1}(\delta - x_{1,1}) \pm O(\delta^4) = \delta^2x_{1,1} + \delta x_{2,1} - 2x_{1,1}x_{2,1}
\pm O(\delta^4).$$
Combining all three estimates, we have
$$N\Big(\emph{\Pr}\big[k_{1} =k_{2} = 1\big] + 2\cdot \emph{\Pr}
  \big[k_{1} = 2, k_{2} = 0\big]\Big) - \delta\cdot 
  \emph{\Pr}\big[k_{1} = 1, k_{2} = 0\big]\approx x_{2,1} \pm O(\delta^3).$$
Since $x_{2,1}\le \delta^2+\lambda$, the linear form on the LHS gives 
  us an additive approximation of $x_{2,1}$.
\end{remark2}

We need some notation in order to generalize and formalize this.
Let $\calS=[n]\setminus \{\ell\}$, the set of players 
  observed by player $P_\ell$. Let $k_b$, $b\in [3]$, denote the random variable that counts
  players from $\calS$ that play strategy $b$.  
We write $\calL=\{i\in \calS:i\le r\}$ and $m=|\calL|$, i.e.,
$\calL = [r]$ and $m=r$ if $\ell > r$, 
  and $\calL = [r] \setminus \{\ell\}$ and $m=r-1$ if $\ell < r$.
We start by understanding the following probabilities  
$$\Big\{\Pr\big[k_{1} = m -j, k_{2} = j\big] : j \in [0:m]\Big\}.$$
It will become clear that players from $\calS\hspace{-0.03cm}\setminus\hspace{-0.03cm}\calL$
  have probabilities too small to significantly affect these probabilities
  (so their contribution will just be absorbed into the error term).

For $j\in [0:m]$, let $\Delta_j$ denote the set of partitions
  of $\calS$ into sets of size $m-j,j$ and $n-1-m$:
$$
\Delta_{j} = \Big\{(\calS_1,\calS_2,\calS_3):\text{$\calS_1,\calS_2,\calS_3$
  are pairwise disjoint,}\  
\calS_1 \cup \calS_2 \cup \calS_3 = \calS,\hspace{0.04cm} 
|\calS_1| = m - j,\hspace{0.04cm} |\calS_2| = j\Big\}.$$ 
So, by definition, we have 
$$ 
\Pr\big[k_{1} = m -j, k_{2} = j\big] = \sum_{(\calS_1,\calS_2,\calS_3) 
\in \Delta_{j}} \left(\prod_{i \in \calS_1} x_{i,1} \prod_{i \in \calS_2} 
x_{i,2} \prod_{i \in \calS_3} x_{i,3}\right). 
$$
By (\ref{eq:setup}) we can
  write $x_{i,1}+x_{i,2}=\delta^i+\lambda_i$ for some 
  $\lambda_i$ with $|\lambda_i|\le \lambda$.
We can substitute to get
\begin{equation}\label{eq:11}
\Pr\big[k_{1} = m -j, k_{2} = j\big] = \sum_{(\calS_1,\calS_2,\calS_3) 
\in \Delta_{j}}\left( \prod_{i \in \calS_1} x_{i,1} \prod_{i \in \calS_2}
(\delta^i + \lambda_i - x_{i,1}) \prod_{i \in \calS_3} 
(1 - \delta^i - \lambda_i)\right).\vspace{0.06cm}
\end{equation}

Next, we split $\Delta_j$ into two sets $\Delta_{j}^*$ and $\Delta_j'$: 
  $(\calS_1,\calS_2,\calS_3)\in \Delta_j$ is in $\Delta_j^*$ if $\calS_1\cup\calS_2
  =\calL$; otherwise, it is in $\Delta_j'$.
This splits the sum in (\ref{eq:11}) into two sums accordingly, one over
  $\Delta_j^*$ and one over $\Delta_j'$.
We show in the following lemma that the contribution from the 
  second sum is negligible.

\begin{lemma}\label{error1}
Given the parameters in \emph{(\ref{eq:setup})}, we have
$$ 
\sum_{(\calS_1,\calS_2,\calS_3) \in \Delta_{j}'}\left( 
\prod_{i \in \calS_1} x_{i,1} \prod_{i \in \calS_2} (\delta^i+\lambda_i- x_{i,1}) 
\prod_{i \in \calS_3} (1 - \delta^i-\lambda_i)\right) = O\bigg(\delta\prod_{i\in \calL}
  \delta^i\bigg).
$$
\end{lemma}
\begin{proof}
Since all terms in the sum are nonnegative, it suffices to show that
\begin{equation}\label{eq:22}
\sum_{(\calS_1,\calS_2,\calS_3) \in \Delta_{j}'}\left( 
\prod_{i \in \calS_1} x_{i,1} \prod_{i \in \calS_2} (\delta^i+\lambda_i- x_{i,1})\right) = O\bigg(\delta\prod_{i\in \calL}
  \delta^i\bigg).
\end{equation}

Fix a set $\calT \subseteq \calS$ such that $|\calT| = m$ but $\calT \neq \calL$.  
We have 
$$ 
\prod_{i\in \calT} (\delta^i+\lambda_i)=
\prod_{i\in \calT} \big(x_{i,1} + (\delta^i+\lambda_i - x_{i,1})\big) 
= \sum_ {\calS_1 \subseteq \calT}\left( \prod_{i \in \calS_1} x_{i,1} 
\prod_{i \in \calT \setminus \calS_1} (\delta^i+\lambda_i - x_{i,1})\right). $$
Since every term on the RHS is nonnegative, we have
$$
\sum_{\substack{\calS_1 \subseteq \calT\\ |\calS_1| = m - j}}\left( 
\prod_{i \in \calS_1} x_{i,1} \prod_{i \in \calT \setminus \calS_1} 
(\delta^i+\lambda_i - x_{i,1}) \right)\leq \prod_{i \in \calT} (\delta^i+\lambda_i)
=(1+o(1))\cdot \prod_{i\in \calT} \delta^i,
$$
given that $\lambda_i=\delta^{n^2}$ in (\ref{eq:setup}).
Let $h(\calT)=\prod_{i\in \calT} \delta^i$.
To prove (\ref{eq:22}), it now suffices to show that
$$
\sum_{\substack{\calT\subseteq \calS\\ |\calT|=m,\hspace{0.03cm} \calT\ne
  \calL}} h(\calT)=
O\big(\delta\cdot h(\calL)\big)=
O\bigg(\delta\prod_{i\in \calL}\delta^i\bigg).
$$

For this purpose, notice that $h(\calT)\le \delta\cdot h(\calL)$ for any $\calT$ such
  that $\calT\subseteq \calS$, $|\calT|=m$, but $\calT\ne \calL$.
It is also easy to see that there is at most one 
  $\calT$ such that $h(\calT)=\delta\cdot h(\calL)$.
Because every other $\calT$ has $h(\calT)\le \delta^2\cdot h(\calL)$
  and the total number of $\calT$'s is at most $2^{n-1}=N/2$, we have
$$
\sum_{\calT} h(\calT)\le \delta\cdot h(\calL)+(N/2)\cdot \delta^2\cdot h(\calL)
=O(\delta\cdot h(\calL)),
$$  
as $\delta=1/N$. This finishes the proof of the lemma.
\end{proof}

Combining (\ref{eq:11}) and Lemma \ref{error1}, we have
$$
{\Pr}\big[k_{1} = m -j, k_{2} = j\big] =\hspace{-0.1cm}
\sum_{\substack{\calS_1 \subseteq \calL\\ |\calS_1| = m-j}} \hspace{-0.1cm}
\left(\prod_{i \in \calS_1} x_{i,1} \prod_{i \in \calL \setminus \calS_1} (\delta^i+\lambda_i - x_{i,1}) \prod_{i \notin \calL}\hspace{0.08cm}(1 - \delta^i-\lambda_i)\right)\hspace{-0.04cm}\pm O(\delta\cdot h(\calL)).
$$
The next lemma further simplifies this estimate by
  absorbing all the $\lambda_i$'s into the error term.
  
\begin{lemma} \label{lem:error2}
Given the parameters in \emph{(\ref{eq:setup})}, we have
$$ 
\emph{\Pr}\big[k_{1} = m -j, k_{2} = j\big] = 
\sum_{\substack{\calS_1\in \calL\\ |\calS_1|=m-j}}\left( 
\prod_{i \in \calS_1} x_{i,1} \prod_{i \in \calL\setminus \calS_1} (\delta^i- x_{i,1}) 
\prod_{i \notin \calL} (1 - \delta^i)\right) 
\pm O(\delta\cdot h(\calL)).
$$
\end{lemma}
\begin{proof}
First the number of $\calS_1$'s is at most $2^{n-1}<N$. Further, 
  fixing an $\calS_1$ and multiplying out 
$$
\prod_{i \in \calS_1} x_{i,1} \prod_{i \in \calL\setminus\calS_1} 
(\delta^i + \lambda_i - x_{i,1})\prod_{i \notin\calL} (1 - \delta^i - \lambda_i) $$ will yield $3^j\cdot 3^{n-1-m}
  \le 3^{n-1}< N^2$ many terms. 
The absolute value of each term with at least
  one $\lambda_i$ must be less than or equal to $\lambda$ because all factors 
  are less than or equal to $1$. There are at most $N^2$ many such terms, for each 
  $\calS_1$, and there are at most $N$ different $\calS_1$'s. 
Using $N^3\lambda\ll \delta h(\calL)$ by (\ref{eq:setup}), we can absorb all terms with at least one 
    $\lambda_i$ into the error term $O\hspace{0.02cm}(\delta\cdot h(\calL))$.
\end{proof}

Using Lemma \ref{lem:error2} and the fact that 
  $\prod_{i\notin \calL}(1-\delta^i)>1/2$ as $\delta=1/2^n$, we have
$$\left(\prod_{i \notin \calL} (1 - \delta^i)\right)^{-1} 
\Pr\big[k_{1} = m - j, k_{2} = j\big] = 
\sum_{\substack{\calS_1 \subseteq \calL\\ |\calS_1| = m-j}} \left( 
\prod_{i \in \calS_1} x_{i,1} \prod_{i \in \calL \setminus \calS_1} 
(\delta^i - x_{i,1})\right)\pm O(\delta\cdot h(\calL)). $$ 
To understand the RHS better, we define a polynomial $P_d^m$ for each $d\in [0:m]$ to be
$$
P_d^{m} = \sum_{\calT \subseteq \calL,\hspace{0.04cm} |\calT| = d} 
\left(\prod_{i \in \calT} x_{i,1} \prod_{i \in \calL \setminus \calT} \delta^i\right),
$$
and prove the following lemma that establishes a connection between them.

\begin{lemma}  \label{lem:lin_eq}
Given $P_d^m$ defined above, we have
\begin{equation}\label{compare}
\sum_{\substack{\calS_1 \subseteq \calL\\ 
|\calS_1| = m-j}}\left( \prod_{i \in \calS_1} x_{i,1} \prod_{i \in \calL \setminus \calS_1} (\delta^i - x_{i,1})\right)  =  \sum_{i = 0}^j 
\hspace{0.05cm}(-1)^i \cdot {m - j + i \choose m - j} \cdot P_{m - j + i}^{m} 
\end{equation}
\end{lemma}

\begin{proof}
Note that every monomial that appears on the two sides of (\ref{compare})
  has the form $\prod_{i\in \calT} x_{i,1}$ for some
  $\calT \subseteq \calL$ with $|\calT| =d\ge  m - j$.
Fix such a $\calT$.
The coefficient of $\prod_{i\in \calT} x_{i,1}$ on RHS of (\ref{compare}) is
$$
(-1)^{d-m+j}\cdot {d\choose m-j}\cdot \prod_{i\in \calL\setminus \calT}\delta^i.
$$
On the other hand,  
for an $\calS_1\subseteq \calL$ with $|\calS_1|=m-j$, we have 
$$ \prod_{i \in \calS_1} x_{i,1} \prod_{i \in\calL \setminus \calS_1} 
(\delta^i - x_{i,1}) = \sum_{\calS' \subseteq \calL \setminus \calS_1}\left( \prod_{i \in \calS_1} x_{i,1} \prod_{i \in \calS'} (-x_{i,1}) \prod_{i \in \calL \setminus \{\calS_1 \cup \calS'\}} \delta^i \right).$$ 
Hence, $\prod_{i \in \calT} x_{i,1} $ occurs exactly once in this sum 
  if and only if $\calS_1 \subseteq \calT$, and will take the form $$ \prod_{i \in \calS_1} x_{i,1} \prod_{i \in \calT \setminus \calS_1} (-x_{i,1})\prod_{i \in 
  \calL \setminus \calT} \delta^i = (-1)^{d-m+j}\prod_{i \in \calT} x_{i,1} \prod_{i \in \calL\setminus \calT}\delta^i.$$
Further, there are ${d \choose m - j}$ many $\calS_1$ such that 
  $\calS_1 \subseteq \calT$ and $|\calS_1|=m-j$.
The lemma is proven.
\end{proof}

Combining Lemma \ref{lem:error2} and \ref{lem:lin_eq},
  we immediately get the following corollary:

\begin{corollary}
For any $j\in [0:m]$, we have \label{coro:lin_eq}
$$\left(\prod_{i \notin \calL} (1 - \delta^i)\right)^{-1} 
\emph{\Pr}\big[k_{1} = m - j, k_{2} = j\big] = 
\sum_{i = 0}^j \hspace{0.05cm}(-1)^i \cdot {m - j + i \choose m - j} \cdot
P_{m- j + i}^{m} \pm O\big(\delta\cdot h(\calL)\big).$$
\end{corollary}

Taking a step back, we have derived a set of linear equations 
  that hold with high precision over 
$\Pr[k_{1} = m, k_{2} = 0],\ldots,\Pr[k_{1} = 0, k_{2} = m]$ and $P_m^{m},\ldots,P_0^{m}$. This then allows us to attain a close approximation for $P_1^{m}$, using a linear form of the $m$ probabilities.
Note that 
\begin{equation}\label{final2}
P_1^{m} = \sum_{i \in \calL}\hspace{0.03cm} x_{i,1} 
\prod_{j \in \calL \setminus \{i\}} \delta^j
=h(\calL)\cdot \sum_{i\in \calL} N^i\cdot x_{i,1}
\end{equation}
is a linear form of the $x_{i,1}$'s, $i\in \calL$, including 
  $x_{r,1}$ (recall that $r$ is the largest integer in $\calL$). 
So from here, it will be straightforward to get 
  an approximation of $x_{r,1}$. 

The next lemma gives us a linear form to approximate $P_1^m$.

\begin{lemma}\label{lem:inverse}
The $m$ probabilities and $P_1^m$ satisfy
 $$ \left(\prod_{i \notin \calL} (1-\delta^i)\right)^{-1}
 \hspace{0.1cm} \sum_{j=1}^{m} j\cdot 
  {\emph{\Pr}}\big[k_{1} = j, k_{2} = m-j\big] = P_1^{m} 
  \pm O\big(m^2\delta\cdot h(\calL)\big).$$
\end{lemma}

\begin{proof}
By Corollary \ref{coro:lin_eq} (and replacing $j$ init by $m-j$), we see that it suffices to show that
\begin{equation}\label{eq:uu} 
P_1^{m} = \sum_{j=1}^{m} j\cdot \left(\sum_{i=0}^{m-j} 
  (-1)^i\cdot  {j + i \choose j}\cdot P_{j+i}^{m}\right). 
\end{equation}
Consider $P_d^{m}$ for some $d \in [m]$. $P_d^{m}$ appears
  in the $j$th term on the RHS of (\ref{eq:uu}) if  
  and only if $d \geq j$, and when this is the case, 
  the coefficient of $P_d^m$ is 
$${ j}\cdot (-1)^{d - j} \cdot {d \choose  j}.\vspace{-0.15cm}$$
So the RHS of (\ref{eq:uu}) is
$$ \sum_{d=1}^m P_d^{m}\cdot \left(
\sum_{j=1}^d (-1)^{d-j}\cdot j\cdot { d \choose j}\right). 
$$
For $d = 1$, the coefficient of $P_1^m$ is clearly $1$. 
For $d>1$, using $j {d \choose j } = d {d - 1 \choose j - 1}$ we have  
$$ 
\sum_{j=1}^d (-1)^{d-j}\cdot { j }\cdot { d \choose j} = 
d\cdot \sum_{j=1}^d (-1)^{d-j}\cdot {d - 1 \choose j - 1} = 
d\cdot \sum_{j=0}^{d-1} (-1)^{d-j- 1} { d - 1 \choose j} = 0. \vspace{0.06cm}$$
This finishes the proof of the lemma.
\end{proof} 

Lemma \ref{lem:inverse} gives us a linear form to approximate $P_1^m$.
Denote this linear form by $Y_m$.
Then for the special case when $\calL=\{r\}$ (so $r$ is the only integer in $\calL$), 
  we are done since
  $P_1^m$ is exactly $x_{r,1}$, and we have attained a linear form
  that approximates $x_{r,1}$ with error $O(m^2\delta\cdot h(\calL))$.

Otherwise suppose $|\calL|>1$.
We use $r'$ to denote the largest integer in $\calL$ other than $r$
  and write $\calL'=\{i\in \calS:i\le r'\}$ ($|\calL'|=m-1$).
Repeating the same line of proof so far over $\calL'$ and
  $m-1$, we obtain a linear form of 
  $\Pr [k_1=m-1-j,k_2=j]$, $j\in [0:m-1]$, denoted by $Y_{m-1}$,
  to approximate
\begin{equation}\label{final1}
P_1^{m-1}=\sum_{i\in \calL'} x_{i,1}\prod_{j\in \calL'\setminus \{i\}}
  \delta^j=h(\calL')\cdot \sum_{i\in \calL'} N^i\cdot x_{i,1}
\end{equation}
with error $O(m^2\delta\cdot h(\calL'))$.
By the definition of $P_1^m$ and $P_1^{m-1}$ in (\ref{final2}) and
  (\ref{final1}), we have
$$
x_{r,1}=\delta^r\left(\frac{P_1^m}{h(\calL)}-\frac{P_1^{m-1}}{h(\calL')}\right).
$$
As a result, we have obtained a linear form
\begin{equation}\label{fffff}
\delta^r\left(\frac{Y_m}{h(\calL)}-\frac{Y_{m-1}}{h(\calL')}\right)
=x_{r,1}\pm O\hspace{0.03cm}(m^2\delta^{r+1})
\end{equation}
over $\Pr[k_{1} = m - j, k_{2} = j]$, $j \in [0:m]$
and $\Pr[k_{,1} = m - 1 - j, k_{2} = j]$, $j \in [0:m-1]$.

Finally, it follows easily from our derivation of $Y_m$ and $Y_{m-1}$
  that coefficients of this linear form can be computed
  in polynomial time in $n$, and every coefficient has
  absolute value at most $N^{m^2}$. 

\def\size{\textsc{size}} \def\aa{\mathbf{a}}

\section{Membership in PPAD}

In this section we show that $(\alpha,c)$-\anonymous\ is in PPAD for
  any constants $\alpha\in \mathbb{N}$ and $c>0$, i.e.
  the problem of finding an $\epsilon$-approximate equilibrium
  in an anonymous game $\calG = (n,\alpha,\{\pay_p\})$
  with payoffs from $[0,1]$ is in PPAD, where $\epsilon=1/2^{n^c}$.
Below we use $\size(\calG)$ to denote the input size
  of an anonymous game $\calG$, i.e., length of the binary representation
  of $\calG$.
We write $\size(a)$ to denote the length of the binary representation
  of a rational number $a$, and let $\size(\aa)=\sum_{i} \size(a_i)$ 
  for a rational vector $\aa$ (e.g., a rational mixed strategy profile).
  
Fix constants $\alpha\in \mathbb{N}$ and $c>0$.
We show the membership of $(\alpha,c)$-\anonymous\ by
  reducing~it to a ``weak-approximation'' fixed point problem \cite{FIXP}
  (see \cite{FIXP} for the difference between \emph{weak} 
  and \emph{strong} approximations).
Given $\calG=(n,\alpha,\{\pay_p\})$, we define
  a map $F:\Delta\rightarrow \Delta$ (this is the map commonly used 
  to prove the existence of Nash equilibria, e.g., see \cite{NASH51}), where
$$
\Delta=\Big\{(\xx_i:i\in [n]): \xx_i\in \mathbb{R}_+^\alpha\ 
  \text{is a mixed strategy of player $i\in [n]$}\Big\}
$$
is the set of all mixed strategy profiles.
For each $i \in [n]$ and $j \in [\alpha]$, the $(i,j)$th component of $F$ 
\begin{equation}\label{defF}
F_{i,j} (\bxx) = \frac{x_{i,j} + \max{(0, u_i(j,\bxx) - u_i(\bxx))}}{1 + \sum_{k \in [\alpha]}\max{(0, u_i(k,\bxx) - u_i(\bxx))}},
\end{equation}
where $\calX=(\xx_i:i\in [n])\in \Delta$ and $\xx_i=(x_{i,1},\ldots,x_{i,\alpha})$
  for each $i\in [n]$.

Observe that $F$ is continuous and maps $\Delta$ to itself. 
We also have 

\begin{property}\label{poly-comp}
The map $F$ defined above is \emph{polynomial-time computable}: 
Given a rational $\bxx\in \Delta$, $F(\bxx)$ is rational and
  can be computed in polynomial time in $\size(\calG)$ and $\size(\bxx)$.
\end{property}

\begin{proof}
This follows from the fact that there is a polynomial-time 
  dynamic programming algorithm (see \cite{dask14}) that computes $u_i(j,\bxx)$,
  given $\calG$ and $\calX$.
\end{proof}

We say $\bxx\in \Delta$ is an \emph{$\epsilon$-approximate fixed point} of $F$ if 
  $\|F(\bxx)-\bxx\|_{\infty}\le \epsilon$.
We prove Lemma \ref{lem:fixedpoint} in Section \ref{sec:fixedpoint},
  showing that approximate fixed points
  of $F$ are approximate Nash equilibria of $\calG$.
  
\begin{lemma}\label{lem:fixedpoint}
Given $\bxx \in \Delta$ and $0 \leq \epsilon \le 1$, if $\|F(\bxx) - \bxx\|_\infty \leq \epsilon$, then
  we have $u_i(j,\bxx) \leq u_i(\bxx) + \epsilon'$ for all players
  $i \in [n]$ and pure strategies $j \in [\alpha]$, where $\epsilon' = \alpha^2\epsilon^{1/3}$.
\end{lemma}

So to find an $\epsilon$-approximate Nash equilibrium $\calX$ of $\calG$,
  it suffices to find an $(\epsilon^3/\alpha^6)$-approximate fixed point of $F$.
Moreover, we show in Section \ref{sec:Lipschitz} that $F$ is \emph{polynomially Lipschitz continuous}:

\begin{lemma}\label{poly-cont}
For all $\bxx , \byy\in \Delta$, we have
$$ \|F(\bxx) - F(\byy)\|_\infty \leq 10\hspace{0.03cm}n 
  \hspace{0.03cm}\alpha^{n+2}\cdot \| \bxx - \byy \|_\infty. $$
\end{lemma}

Combining Property \ref{poly-comp} and Lemma \ref{poly-cont},
  it follows from Proposition 2.2 (Part 2) of \cite{FIXP} that
  given $\calG$ and $\epsilon$ (in binary), the problem of finding an 
  $\epsilon$-approximate fixed point $\bxx$ of $F$ is in PPAD. 
The PPAD membership of $(\alpha,c)$-\anonymous\ 
  then follows from Lemma \ref{lem:fixedpoint}.

\subsection{Proof of Lemma \ref{lem:fixedpoint}}\label{sec:fixedpoint}

For convenience, we write $\max_{i,k}(\bxx) = \max{(0, u_i(k,\bxx) - u_i(\bxx))}$
  for $i\in [n]$ and $k\in [\alpha]$.

In the pursuit of a contradiction, assume that there exist a player $i\in [n]$ 
  and an action $\ell\in [\alpha]$ such that $u_i(\ell,\bxx)> u_i(\bxx) + \epsilon'$. This, along with the fact that $\max_{i,k}(\bxx) \in [0,1]$, implies that,
\begin{eqnarray}\label{sum:ineq}
&\epsilon' < \sum_{k \in [\alpha]} \max_{i,k}(\bxx) \leq \alpha-1. &
\end{eqnarray}
  
We will show that cases $x_{i,\ell} \leq \alpha\hspace{0.02cm}\epsilon^{1/3}$ 
  and $x_{i,\ell} >\alpha\hspace{0.02cm}\epsilon^{1/3}$ both result 
  in the existence of a strategy $j \in [\alpha]$ such that 
  $|F_{i,j}(\bxx) - x_{i,j}| > \epsilon$, contradicting our 
  initial assumption.\vspace{0.07cm}
\begin{flushleft}\begin{enumerate} 
\item[] {Case 1}: $x_{i,\ell} \leq \alpha\hspace{0.02cm}\epsilon^{1/3}$. Apply (\ref{defF}),
  (\ref{sum:ineq}) and $\epsilon'=\alpha^2\epsilon^{1/3}$ to get
$$
F_{i,\ell}(\bxx) > \frac{x_{i,\ell} + \epsilon'}{\alpha}\ \Rightarrow\ 
F_{i,\ell}(\bxx) - x_{i,\ell} > \frac{\epsilon' - (\alpha - 1) x_{i,\ell}}{\alpha} \geq \frac{\epsilon' - (\alpha - 1) \alpha\hspace{0.02cm}\epsilon^{1/3}}{\alpha} = \epsilon^{1/3}\ge \epsilon. 
$$ 

\item[] {Case 2}: $x_{i,\ell} > \alpha\hspace{0.02cm}\epsilon^{1/3}$.
Let $J = \{j\in [\alpha]:u_i(j,\bxx) \leq u_i(\bxx)\}$. We must have 
\begin{eqnarray*}
&\sum_{j \in J} x_{i,j}\big(u_i(\bxx) - u_i(j,\bxx)\big) \geq 
  x_{i,\ell}\big(u_i(\ell,\bxx) - u_i(\bxx)\big),&
\end{eqnarray*}
where $u_i(\bxx) - u_i(j,\bxx) \leq 1 - \epsilon'$, $u_i(\ell,\bxx) - u_i(\bxx) \geq \epsilon'$, and $x_{i,\ell} > \alpha\hspace{0.02cm}\epsilon^{1/3}$.
Therefore, 
\begin{eqnarray*}
&\sum_{j \in J} x_{i,j}&\hspace{-0.22cm} \geq  
\frac{\alpha\hspace{0.02cm}\epsilon'\epsilon^{1/3}}{1 - \epsilon'},
\end{eqnarray*} 
which implies that there exists some strategy $j \in J$ such that $x_{i,j} \geq 
  \epsilon'\epsilon^{1/3}/(1 - \epsilon')$. Apply (\ref{defF}) and 
  (\ref{sum:ineq}) to get  $F_{i,j}(\bxx) < x_{i,j}/(1 + \epsilon')$, which
  implies that $$\big|F_{i,j}(\bxx) - x_{i,j}\big| > \frac{\epsilon' x_{i,j}}{1 + \epsilon'} \geq \frac{(\epsilon')^2\epsilon^{1/3}}{ (1 - \epsilon')(1 + \epsilon')} \geq \alpha^4 
  \epsilon\geq \epsilon.$$
\end{enumerate}\end{flushleft}

This finishes the proof of Lemma \ref{lem:fixedpoint}.

\subsection{Proof of Lemma \ref{poly-cont}}\label{sec:Lipschitz}

As $\bxx - \byy$ is of length $n\hspace{0.02cm}\alpha$, 
  we have $\|\bxx-\byy\|_1\le n\hspace{0.02cm}\alpha\cdot \|\bxx-\byy\|_\infty$.
Thus, it suffices to show that  
$$\| F(\bxx) - F(\byy)\|_\infty \leq 16\hspace{0.02cm}\alpha^{n+1}\| \bxx - \byy\|_1.$$

Fix $i\in [n]$ and $j\in [\alpha]$. We have\vspace{0.04cm}
$$ 
\big|\hspace{0.02cm} F_{i,j}(\bxx) - F_{i,j}(\byy) \hspace{0.02cm}\big| = 
\left|\hspace{0.06cm} \frac{x_{i,j} + \max_{i,j}(\bxx)}{1 + \sum_{k \in [\alpha]}\max_{i,k}(\bxx)} - \frac{y_{i,j} + \max_{i,j}(\byy)}{1 + \sum_{k \in [\alpha]}\max_{i,k}(\byy)} \hspace{0.06cm}\right|.
\vspace{0.04cm} 
$$ 
Multiplying the terms in the RHS to get a common denominator, which is clearly $ \geq 1$, we get
\begin{align}\label{main-bound}
\big|\hspace{0.02cm} F_{i,j}(\bxx) - F_{i,j}(\byy) \hspace{0.02cm}\big| &\leq \big|x_{i,j} - y_{i,j}\big| + \Bigg|\hspace{0.04cm} x_{i,j} \sum_{k \in [\alpha]}\max_{i,k}(\byy) - y_{i,j} \sum_{k \in [\alpha]}\max_{i,k}(\bxx)\hspace{0.04cm} \Bigg| \\[0.5ex] &
\ \ \ + \Big| \max_{i,j}(\bxx) - \max_{i,j}(\byy)\hspace{0.02cm} \Big| + \Bigg|\hspace{0.01cm} \max_{i,j}(\bxx) \sum_{k \in [\alpha]}\max_{i,k}(\byy) - \max_{i,j}(\byy) \sum_{k \in [\alpha]}\max_{i,k}(\bxx) \hspace{0.04cm}\Bigg|.
\nonumber
\end{align}

To bound $| F_{i,j}(\bxx) - F_{i,j}(\byy)|$, we shall use the following 
  simple trick several times in the rest of the proof.
If $a_1,a_2,b_1,b_2 \in [0,1]$, then we have 
$$|\hspace{0.02cm}a_1a_2 - b_1b_2\hspace{0.02cm}| = |\hspace{0.02cm}(a_1 - b_1)a_2 + b_1(a_2 - b_2)\hspace{0.02cm}| \leq |\hspace{0.02cm}a_1 - b_1\hspace{0.02cm}| + |\hspace{0.02cm}a_2 - b_2\hspace{0.02cm}|, $$ which easily extends to $$ |\hspace{0.02cm}a_1 \cdots a_n - b_1 \cdots b_n\hspace{0.02cm}| \leq |\hspace{0.02cm}a_1 - b_1\hspace{0.02cm}| + \cdots + |\hspace{0.02cm}a_n - b_n\hspace{0.02cm}|, $$ when all the $a_i$'s and $b_i$'s are in $[0,1]$.

Now we come back to (\ref{main-bound}).
By the definition of $\max_{i,j}(\bxx)$, we have
\begin{align*}  
\Big| \max_{i,j}(\bxx) - \max_{i,j}(\byy)\hspace{0.02cm} \Big| &\leq \big|(u_i(j,\bxx) - u_i(\bxx)) - (u_i(j,\byy) - u_i(\byy))\big|\\ &\leq \big|u_i(j,\bxx) - u_i(j,\byy)\big| + \big|u_i(\bxx) - u_i(\byy)\big|.
\end{align*}
As $\bxx,\byy\in \Delta$ we have $x_{i,j},y_{i,j} \in [0,1]$. 
Since all payoffs of $\calG$ are in $[0,1]$, we have $u_i(j,\bxx),u_i(j,\byy)$, $u_i(\bxx),u_i(\byy) \in [0,1]$ for all $i,j$, which in turn implies that $\max_{i,j}(\bxx),\max_{i,j}(\byy) \in [0,1]$. 

Using these properties above, along with the trick, we can conclude 
\begin{align*}
\Bigg|\hspace{0.04cm} x_{i,j} \sum_{k \in [\alpha]}\max_{i,k}(\byy) &- y_{i,j} \sum_{k \in [\alpha]}\max_{i,k}(\bxx)\hspace{0.04cm} \Bigg| \leq \sum_{k \in [\alpha]}\Big|\hspace{0.03cm}x_{i,j}\cdot \max_{i,k}(\byy) - y_{i,j}\cdot \max_{i,k}(\bxx)\hspace{0.03cm}\Big| \\[0.8ex]
&\leq \sum_{k\in [\alpha]}\Big(|x_{i,j} - y_{i,j}| + |u_i(k,\bxx) - u_i(k,\byy)| +|u_i(\bxx) - u_i(\byy)|\Big).
\end{align*}
Similarly, we also have 
\begin{align*}
&\Bigg|\max_{i,j}(\bxx) \sum_{k \in [\alpha]}\max_{i,k}(\byy) - \max_{i,j}(\byy) \sum_{k \in [\alpha]}\max_{i,k}(\bxx)\hspace{0.04cm} \Bigg|\\[0.8ex] &\ \ \ \ \ \ \ \ \ \ \ \ \ \ \leq \sum_{k\in [\alpha]}\Big(|u_i(j,\bxx) - u_i(j,\byy)| + 2|u_i(\bxx) - u_i(\byy)| + |u_i(k,\bxx) - u_i(k,\byy)| \Big).
\end{align*}

Plugging all these back into (\ref{main-bound}), we have
\begin{align*}
\big| F_{i,j}(\bxx) - F_{i,j}(\byy) \big|&\leq (1 + \alpha)\cdot |\hspace{0.02cm}x_{i,j}-y_{i,j}\hspace{0.02cm}| + 
(1+3\alpha)\cdot |\hspace{0.02cm}u_i(\bxx) - u_i(\byy)\hspace{0.02cm}|\\&\ \ \ \ \ \ + (1+\alpha)\cdot|\hspace{0.02cm}u_i(j,\bxx) - u_i(j,\byy)\hspace{0.02cm}|+ 
2\cdot \sum_{k \in [\alpha]}|\hspace{0.02cm}u_i(k,\bxx) - u_i(k,\byy)\hspace{0.02cm}|.
\end{align*}

Finally, we bound $|u_i(k,\bxx) - u_i(k,\byy)|$ in terms of $\|\bxx - \byy\|_1$. Let $S$ be the set of pure strategy profiles. Then, by applying the trick and the fact that all payoffs are in $[0,1]$, it follows that \vspace{0.08cm}
\begin{align*}
|\hspace{0.02cm}u_i(k,\bxx) - u_i(k,\byy)\hspace{0.02cm}| &\leq \sum_{\ss \in S_{-i}} 
\Bigg|\prod_{q \neq i} x_{q,s_q} - \prod_{q \neq i} y_{q,s_q}\Bigg| 
\leq \sum_{\ss \in S_{-i}} \sum_{q \neq i} |\hspace{0.02cm}x_{q,s_q} - y_{q,s_q}\hspace{0.02cm}| \leq \alpha^{n-1}\| \bxx - \byy\|_1 
\\[1ex] |\hspace{0.02cm} u_i(\bxx) - u_i(\byy)\hspace{0.02cm} | &\leq \sum_{\ss \in S} \Bigg|\prod_{q\in [n]} x_{q,s_q} - \prod_{q \in [n]} y_{q,s_q}\Bigg| \leq \sum_{\ss \in S} \sum_{q \in [n]} |\hspace{0.02cm}x_{q,s_q} - y_{q,s_q}\hspace{0.02cm}| \leq \alpha^n \| \bxx - \byy \|_1.\\[-2.7ex]
\end{align*}

Applying these inequalities, along with $|\hspace{0.02cm}x_{i,j} - y_{i,j}\hspace{0.02cm}| 
  \leq \|\bxx - \byy\|_1$, we get 
$$ \big|F_{i,j}(\bxx) - F_{i,j}(\byy)\big| \leq 10\hspace{0.02cm}
  \alpha^{n+1}\cdot \| \bxx - \byy \|_1.$$
This finishes the proof Lemma \ref{poly-cont}.

\section{Open Problems}

Can the number of strategies be further reduced from seven in our PPAD-hardness result? Specifically, could we construct an anonymous game similar to the radix game $\calG_{n,N}$, particularly its set of approximate Nash equilibria after perturbation, but without the four special (auxiliary) pure~stra\-tegies $\{q_1,q_2,r_1,r_2\}$? While we believe this to be possible, constructing such a game can be highly non-trivial and would require specifying different payoffs for many of the possible outcomes seen by each player. Accordingly, proving a result similar to Lemma \ref{lem:perturb} after duplicating the first strategy would be even more difficult.

However, even the construction of such a game would only reduce the number of 
  strategies~used in the hardness proof down to three
  (due to the strategy duplication in the generalized radix game later), 
  leading to the next open question: 
Is there an FPTAS for two-strategy anonymous games? As~was posited by Daskalakis and Papadimitriou, it remains unclear whether a rational two-strategy anonymous game always has a rational Nash equilibrium. Additionally, in their sequence of paper's proving a PTAS for a bounded number of strategies, Daskalakis and Papadimitriou found that the form of the $\Pr_\bxx[p,\kk]$ is significantly simpler for two-strategy anonymous games. Correspondingly, we found that constructing useful gadgets for reductions with just two strategies to be very difficult, suggesting that an FPTAS for two-strategy anonymous games is certainly a possibility.

Moreover, could there be an FPTAS for anonymous games with any bounded number of~pure strategies? There is no clear way to strengthen our current construction to obtain a PPAD-hardness result for ${1}/{\text{poly}(n)}$-approximate Nash equilibrium. In order for the estimation lemma to hold, we need $x_{i,1} + x_{i,2} \approx \delta^i$ for all $i$. So even if we set $N = 2$, ensuring that $x_{i,1} + x_{i,2} = \delta^i \pm O(1/\text{poly}(n))$ would still not be sufficient for the estimation lemma to hold. Accordingly, in order to modify our construction to get such a hardness result, we would need to construct an anonymous game, which contains $n$ players with the same properties as the main players in the generalized radix game, but with the additional property that $O(1/\text{poly}(n))$ shifts in the payoffs would only cause $O(1/2^{\text{poly}(n)})$ shifts in $x_{i,1} + x_{i,2}$, which seems incredibly unlikely. 

\begin{flushleft}
\bibliography{citations,Reference}

\begin{thebibliography}{ABOV07}

\bibitem[ABOV07]{PlanarP}
L.~Addario-Berry, N.~Olver, and A.~Vetta.
\newblock A polynomial time algorithm for finding {N}ash equilibria in planar
  win-lose games.
\newblock {\em Journal of Graph Algorithms and Applications}, 11(1):309--319,
  2007.

\bibitem[AKV05]{AbbottKaneValiant}
T.~Abbott, D.~Kane, and P.~Valiant.
\newblock On the complexity of two-player win-lose games.
\newblock In {\em Proceedings of the 46th Annual IEEE Symposium on Foundations
  of Computer Science}, pages 113--122, 2005.

\bibitem[ARV08]{AckermannPLS}
H.~Ackermann, H.~R\"{o}glin, and B.~V\"{o}cking.
\newblock On the impact of combinatorial structure on congestion games.
\newblock {\em Journal of the ACM}, 55:1--22, 2008.

\bibitem[AS11]{Azrieli}
Y.~Azrieli and E.~Shmaya.
\newblock Lipschitz games.
\newblock {\em Mimeo, Ohio State University}, 2011.

\bibitem[Bab13]{BRE12}
Y.~Babichenko.
\newblock Best-reply dynamics in large binary-choice anonymous games.
\newblock {\em Games and Economic Behavior}, 81:130--144, 2013.

\bibitem[BBM07]{BosseMarkakis}
H.~Bosse, J.~Byrka, and E.~Markakis.
\newblock New algorithms for approximate {Nash} equilibria in bimatrix games.
\newblock In {\em Proceedings of the 3rd International Workshop on Internet and
  Network Economics}, pages 17--29, 2007.

\bibitem[BFH09]{Bra09}
F.~Brandt, F.~Fischer, and M.~Holzer.
\newblock Symmetries and the complexity of pure {N}ash equilibrium.
\newblock {\em Journal of Computer and System Sciences}, 75(3):163--177, 2009.

\bibitem[Blo99]{Blo99}
M.~Blonski.
\newblock Anonymous games with binary actions.
\newblock {\em Games and Economic Behavior}, pages 171--180, 1999.

\bibitem[Blo05]{Blo05}
M.~Blonski.
\newblock The women of {C}airo: Equilibria in large anonymous games.
\newblock {\em Journal of Mathematical Economics}, pages 254--263, 2005.

\bibitem[BVV05]{BaranyVempalaVetta}
I.~B\'ar\'any, S.~Vempala, and A.~Vetta.
\newblock Nash equilibria in random games.
\newblock In {\em Proceedings of the 46th Annual IEEE Symposium on Foundations
  of Computer Science}, pages 123--131, 2005.

\bibitem[CD11]{Cai11}
Y.~Cai and C.~Daskalakis.
\newblock On minmax theorems for multiplayer games.
\newblock In {\em Proceedings of the 22nd ACM-SIAM Symposium on Discrete
  Algorithms}, pages 217--234, 2011.

\bibitem[CDT06]{ChenDengTengSparse}
X.~Chen, X.~Deng, and S.-H. Teng.
\newblock Sparse games are hard.
\newblock In {\em Proceedings of the 2nd Workshop on Internet and Network
  Economics}, pages 262--273, 2006.

\bibitem[CDT09]{2Nash}
X.~Chen, X.~Deng, and S.-H. Teng.
\newblock Settling the complexity of computing two-player {N}ash equilibria.
\newblock {\em Journal of the ACM}, 56(3):1--57, 2009.

\bibitem[CFGS11]{CFGS1}
I.~Caragiannis, A.~Fanelli, N.~Gravin, and A.~Skopalik.
\newblock Efficient computation of approximate pure {N}ash equilibria in
  congestion games.
\newblock In {\em Proceedings of the 52nd Annual IEEE Symposium on Foundations
  of Computer Science}, pages 532--541, 2011.

\bibitem[CFGS12]{CFGS2}
I.~Caragiannis, A.~Fanelli, N.~Gravin, and A.~Skopalik.
\newblock Approximate pure {N}ash equilibria in weighted congestion games:
  existence, efficient computation, and structure.
\newblock In {\em Proceedings of the 13th ACM Conference on Electronic
  Commerce}, 2012.

\bibitem[CPY13]{market}
X.~Chen, D.~Paparas, and M.~Yannakakis.
\newblock The complexity of non-monotone markets.
\newblock In {\em Proceedings of the 45th Annual ACM Symposium on Theory of
  Computing}, pages 181--190, 2013.

\bibitem[CS11]{Chien07}
S.~Chien and A.~Sinclair.
\newblock Convergence to approximate {N}ash equilibria in congestion games.
\newblock {\em Games and Economic Behavior}, 71(2):315--327, 2011.

\bibitem[CTV07]{ChenTengValiant}
X.~Chen, S.-H. Teng, and P.~Valiant.
\newblock The approximation complexity of win-lose games.
\newblock In {\em Proceedings of the 18th Annual ACM-SIAM Symposium on Discrete
  Algorithms}, pages 159--168, 2007.

\bibitem[Das08]{2PTAS08}
C.~Daskalakis.
\newblock An efficient {PTAS} for two-strategy anonymous games.
\newblock In {\em Proceedings of the 4th International Workshop on Internet and
  Network Economics}, pages 186--197, 2008.

\bibitem[DGP09]{DGPJournal}
C.~Daskalakis, P.W. Goldberg, and C.H. Papadimitriou.
\newblock The complexity of computing a {N}ash equilibrium.
\newblock {\em SIAM Journal on Computing}, 39(1), 2009.

\bibitem[DMP06]{DaskalakisMehtaPapa}
C.~Daskalakis, A.~Mehta, and C.H. Papadimitriou.
\newblock A note on approximate {Nash} equilibria.
\newblock In {\em Proceedings of the 2nd Workshop on Internet and Network
  Economics}, pages 297--306, 2006.

\bibitem[DMP07]{DaskalakisMehtaPapa2}
C.~Daskalakis, A.~Mehta, and C.H. Papadimitriou.
\newblock Progress in approximate {Nash} equilibria.
\newblock In {\em Proceedings of the 8th ACM Conference on Electronic
  Commerce}, pages 355--358, 2007.

\bibitem[DP07]{PTAS07}
C.~Daskalakis and C.H. Papadimitriou.
\newblock Computing equilibria in anonymous games.
\newblock In {\em Proceedings of the 48th Annual IEEE Symposium on Foundations
  of Computer Science}, pages 83--93, 2007.

\bibitem[DP08]{cPTAS08}
C.~Daskalakis and C.H. Papadimitriou.
\newblock Discretized multinomial distributions and {N}ash equilibria in
  anonymous games.
\newblock In {\em Proceedings of the 49th Annual IEEE Symposium on Foundations
  of Computer Science}, pages 25--34, 2008.

\bibitem[DP09]{obliv}
C.~Daskalakis and C.H. Papadimitriou.
\newblock On oblivious {PTAS}'s for {N}ash equilibrium.
\newblock In {\em Proceedings of the 41st Annual ACM Symposium on Theory of
  Computing}, pages 75--84, 2009.

\bibitem[DP14]{dask14}
C.~Daskalakis and C.H. Papadimitriou.
\newblock Approximate {N}ash equilibria in anonymous games.
\newblock {\em Journal of Economic Theory}, 2014.

\bibitem[ES05]{Sandholm05}
J.C. Ely and W.H. Sandholm.
\newblock Evolution in {B}ayesian games {I}.
\newblock {\em Theory of Games and Economic Behavior}, 53:83--109, 2005.

\bibitem[EY10]{FIXP}
K.~Etessami and M.~Yannakakis.
\newblock On the complexity of {Nash} equilibria and other fixed points.
\newblock {\em SIAM Journal on Computing}, 39(6):2531--2597, 2010.

\bibitem[FNS07]{DederNazerzadehSaberi}
T.~Feder, H.~Nazerzadeh, and A.~Saberi.
\newblock Approximating {Nash} equilibria using small-support strategies.
\newblock In {\em Proceedings of the 8th ACM Conference on Electronic
  Commerce}, pages 352--354, 2007.

\bibitem[FPT04]{fabrikant04}
A.~Fabrikant, C.H. Papadimitriou, and K.~Talwar.
\newblock The complexity of pure {Nash} equilibria.
\newblock In {\em Proceedings of the 36th Annual ACM Symposium on Theory of
  Computing}, pages 604--612, 2004.

\bibitem[HR04]{NashSocial}
C.A. Holt and A.E. Roth.
\newblock The {N}ash equilibrium: {A} perspective.
\newblock {\em Proceedings of the National Academy of Sciences},
  101(12):3999--4002, 2004.

\bibitem[Kal05]{Kal05}
E.~Kalai.
\newblock Partially-specified large games.
\newblock {\em Proceedings of the 1st International Workshop on Internet and
  Network Economics}, pages 3--13, 2005.

\bibitem[KPS06]{KontogiannisPana}
S.~Kontogiannis, P.~Panagopoulou, and P.~Spirakis.
\newblock Polynomial algorithms for approximating {Nash} equilibria of bimatrix
  games.
\newblock In {\em Proceedings of the 2nd Workshop on Internet and Network
  Economics}, pages 286--296, 2006.

\bibitem[KS07]{KontogiannisSpirakis}
S.C. Kontogiannis and P.G. Spirakis.
\newblock Efficient algorithms for constant well supported approximate
  equilibria in bimatrix games.
\newblock In {\em Proceedings of the 34th International Colloquium on the
  Automata, Languages and Programming}, pages 595--606, 2007.

\bibitem[KT07]{KannanTheobald}
R.~Kannan and T.~Theobald.
\newblock Games of fixed rank: A hierarchy of bimatrix games.
\newblock In {\em Proceedings of the 18th Annual ACM-SIAM Symposium on Discrete
  Algorithms}, pages 1124--1132, 2007.

\bibitem[LMM04]{Lipton}
R.J. Lipton, E.~Markakis, and A.~Mehta.
\newblock Playing large games using simple strategies.
\newblock In {\em Proceedings of the 4th ACM Conference on Electronic
  Commerce}, pages 36--41, 2004.

\bibitem[Meh14]{Mehta14}
R.~Mehta.
\newblock Constant rank bimatrix games are {PPAD}-hard.
\newblock In {\em Proceedings of the 46th Annual ACM Symposium on Theory of
  Computing}, pages 545--554, 2014.

\bibitem[Mil96]{Mil96}
I.~Milchtaich.
\newblock Congestion games with player-specific payoff functions.
\newblock {\em Games and Economic Behavior}, pages 111--124, 1996.

\bibitem[Nas50]{NASH50}
J.F. Nash.
\newblock Equilibrium points in n-person games.
\newblock {\em Proceedings of the National Academy of Sciences}, 36(1):48--49,
  1950.

\bibitem[Nas51]{NASH51}
J.F. Nash.
\newblock Non-cooperative games.
\newblock {\em Annals of Mathematics}, 54(2):286--295, 1951.

\bibitem[PR08]{PapRough08}
C.H. Papadimitriou and T.~Roughgarden.
\newblock Computing correlated equilibria in multi-player games.
\newblock {\em Journal of the ACM}, 55(3):1--29, 2008.

\bibitem[Ras83]{Rashid83}
S.~Rashid.
\newblock Equilibrium points of nonatomic games: Asymptotic results.
\newblock {\em Economics Letters}, 12:7--10, 1983.

\bibitem[Ros73]{ros}
R.W. Rosenthal.
\newblock A class of games possessing pure-strategy {N}ash equilibria.
\newblock {\em International Journal of Game Theory}, 2(1):65--67, 1973.

\bibitem[Sch73]{sch73}
D.~Schmeidler.
\newblock Equilibrium points of nonatomic games.
\newblock {\em Journal of Statistical Physics}, 7(4):295--300, 1973.

\bibitem[SV08]{Skopalik}
A.~Skopalik and B.~V\"{o}cking.
\newblock Inapproximability of pure {N}ash equilibria.
\newblock In {\em Proceedings of the 40th Annual ACM Symposium on Theory of
  Computing}, pages 355--364, 2008.

\bibitem[TS07]{TSepsilon}
H.~Tsaknakis and P.G. Spirakis.
\newblock An optimization approach for approximate {Nash} equilibria.
\newblock In {\em Proceedings of the 3rd International Workshop on Internet and
  Network Economics}, pages 42--56, 2007.

\bibitem[TS10]{TsaknakisS}
H.~Tsaknakis and P.G. Spirakis.
\newblock Practical and efficient approximations of {Nash} equilibria for
  win-lose games based on graph spectra.
\newblock In {\em Proceedings of the 6th International Conference on Internet
  and Network Economics}, pages 378--390, 2010.

\end{thebibliography}
\bibliographystyle{alpha}
\end{flushleft}
\end{document}